\documentclass[a4paper,11pt]{article}

\usepackage[margin=2.5cm]{geometry}

\usepackage{booktabs}
\usepackage{mathrsfs}

\usepackage{amsmath,amssymb,graphicx,amsthm,dsfont}
\usepackage[pagebackref]{hyperref}
\hypersetup{%
  hypertexnames=true,
colorlinks=true,citecolor=green!60!black,linkcolor=blue!40!black%
}

\usepackage{authblk}

\usepackage{multicol}
\usepackage{multirow}
\usepackage{caption}
\usepackage{amsmath}
\usepackage{amssymb}
\usepackage{comment}
\usepackage{amsfonts}

\usepackage{dsfont}

\makeatletter
\newif\if@restonecol
\makeatother

\usepackage[linesnumbered,ruled,vlined]{algorithm2e}
\usepackage{algpseudocode}

\usepackage{tikz}
\usetikzlibrary{positioning,backgrounds,patterns,calc,fit,shapes,matrix}
\tikzstyle{vert}=[circle,draw=black,minimum size=8pt,inner sep=1pt]
\tikzstyle{vertex2}=[circle,draw=black,minimum size=15pt,inner sep=2pt]
\tikzstyle{edge}=[]
\tikzstyle{ypath}=[ultra thick]
\tikzstyle{dottedEdge}=[dotted,thick]
\tikzstyle{small-vertex}=[circle,draw=black,minimum size=6pt,inner sep=0pt,fill=white]
\tikzstyle{thinedges}=[draw=gray!30]
 \tikzstyle{boxes}=[draw,thick, rounded corners=3mm,text width=2.7cm,align=center,text opacity=1,fill opacity=1,fill=white]
\tikzstyle{unk}=[fill=gray!25!white]
\tikzstyle{nodest}=[circle,draw,minimum size=0.55cm,inner sep=.7pt]

\usepackage[margin=0pt,position=t]{subcaption}
\usepackage{mathtools}
\usepackage[sort&compress,nameinlink,capitalize]{cleveref}
\crefname{table}{Table}{Tables}
\crefname{figure}{Figure}{Figures}
\crefname{cor}{Corollary}{Corollaries}
\crefname{step}{Step}{Steps}
\crefname{rrulev}{Rule}{Rules}
\crefname{thm}{Theorem}{Theorems}
\crefname{obs}{Observation}{Observations}
\crefname{lem}{Lemma}{Lemmas}
\crefname{claim}{Claim}{Claims}
\crefname{section}{Section}{Sections}
\crefname{subsection}{Section}{Sections}
\crefname{figure}{Figure}{Figures}
\crefname{algorithm}{Algorithm}{Algorithms}
\crefname{proposition}{Proposition}{Propositions}
\crefname{theorem}{Theorem}{Theorem}
\crefname{lemma}{Lemma}{Lemmas}
\crefname{construction}{Construction}{Constructions}
\crefalias{AlgoLine}{line}

\newtheorem{theorem}{Theorem}
\newtheorem{definition}{Definition}
\newtheorem{lemma}{Lemma}
\newtheorem{corollary}{Corollary}
\newtheorem{example}{Example}
\newtheorem{proposition}{Proposition}
\newtheorem{claim}{Claim}
\newtheorem{construction}{Construction}

\theoremstyle{remark}

\newtheorem{rrule}{Reduction Rule}{\bfseries}{\normalfont}
\crefname{rrule}{Rule}{Rules}

\newtheorem{observation}{Observation}{\bfseries}{\normalfont}
\crefname{observation}{Observation}{Observations}

\newcommand{\woneh}{\textsf{W[1]}-hard}
\newcommand{\xp}{\textsf{XP}}
\newcommand{\np}{\textsf{NP}}
\newcommand{\fpt}{\textsf{FPT}}

\DeclareMathOperator{\pen}{pen}
\DeclareMathOperator{\dist}{dist}

\newcommand{\decprob}[3]{%
  \begin{center}%
    \begin{minipage}{0.9\linewidth}%
      \textsc{#1}\\
      \textbf{Input:} #2\\
      \textbf{Question:} #3
    \end{minipage}%
  \end{center}%
}

\tikzstyle{agent} = [draw, circle, fill=black, minimum size=2ex, inner sep=1pt, text centered, align=center]

\makeatletter
\newcommand{\gettikzxy}[3]{%
  \tikz@scan@one@point\pgfutil@firstofone#1\relax
  \edef#2{\the\pgf@x}%
  \edef#3{\the\pgf@y}%
}
\makeatother

\pgfdeclarelayer{background}
\pgfdeclarelayer{foreground}

\usepackage{cases}
\usepackage{microtype,ellipsis}
\usepackage{color}

\usepackage{units}
\usepackage[linecolor=green!70!black, backgroundcolor=green!10, bordercolor=black, textsize=scriptsize, obeyFinal, disable]{todonotes}

\newcommand{\todohua}[1]{\todo[inline,backgroundcolor=yellow!10]{#1}}

\usepackage{marginnote}

\usepackage[numbers,sort]{natbib}

\makeatletter
\def\NAT@spacechar{~}
\makeatother

\usepackage{colortbl}
\usepackage{xcolor}
\usepackage{paralist}
\usepackage{enumerate}

\usepackage{etoolbox}
\usepackage{appendix}

\newcommand{\appendixproof}[2]{%
    #2
}
\newcommand{\citeGIauthors}{Gusfield and Irving}

\newcommand{\sucw}{\mathsf{succ}}

\newcommand{\agent}{\text{agent}}

\newcommand{\proposalset}{\mathsf{S}}

\newcommand{\SM}{\textsc{Stable Marriage}\xspace}
\newcommand{\ISM}{\textsc{Incremental Stable Marriage}\xspace}

\newcommand{\SR}{\textsc{Stable Roommates}\xspace}
\newcommand{\ISR}{\textsc{Incremental Stable Roommates}\xspace}
\newcommand{\IS}{\textsc{Independent Set}\xspace}

\newcommand{\EIIS}{\textsc{Edge-Incremental Independent Set}\xspace}
\newcommand{\EDClPE}{\textsc{Edge-Decremental Clique with Pendant Edges}\xspace}

\newcommand{\WCFCS}{\textsc{Weighted Conflict-Free Closed Subset}\xspace}

\newcommand{\sdist}{\ensuremath{\mathsf{\delta}}}

\newcommand{\pred}{\ensuremath{\rhd}}
\newcommand{\predr}{\ensuremath{\unrhd}}

\title{Adapting Stable Matchings to Evolving Preferences\thanks{We thank anonymous reviewers of AAAI for their insightful comments.}}

\author[1]{Robert Bredereck}
\author[2]{Jiehua Chen\thanks{Main work done while JC was with University of Warsaw, where she was supported by the European Research Council (ERC) under the European Union's Horizon 2020 research and innovation programme under grant agreement numbers~677651. JC was also supported by the WWTF research project~(VRG18-012).}}
\author[3]{Du\v{s}an Knop\thanks{Main work done while DK was with TU~Berlin, supported by the DFG project MaMu (NI\,369/19).}}
\author[1]{Junjie Luo\thanks{Supported by CAS-DAAD Joint Fellowship Program for Doctoral Students of UCAS and by the DFG project AFFA (BR\,5207/1 and NI\,369/15)}}
\author[1]{Rolf Niedermeier}

\affil[1]{Technische Universit\"at Berlin, Chair of Algorithmics and Computational Complexity
  \texttt{\{robert.bredereck, junjie.luo, rolf.niedermeier\}@tu-berlin.de}}
\affil[2]{Algorithms and Complexity Group, TU Wien, Vienna, Austria
  \texttt{jiehua.chen2@gmail.com}}
\affil[3]{Department of Theoretical Computer Science, Faculty of Information Technology, Czech Technical University in Prague, Prague, Czech Republic
  \texttt{dusan.knop@fit.cvut.cz}}

\begin{document}

\maketitle
\begin{abstract}
Adaptivity to changing environments and constraints is key to success
in modern society. We address this by proposing
``incrementalized versions'' of \textsc{Stable Marriage} and
\textsc{Stable Roommates}. That is, we try to answer the following question:
for both problems, what is the computational cost of adapting an existing stable matching
after some of the preferences of the agents have changed. While doing so,
we also model the constraint that the new stable matching shall be not too different from the old one.
After formalizing these incremental versions, we provide a fairly comprehensive
picture of the computational complexity landscape of
\textsc{Incremental Stable Marriage} and
\textsc{Incremental Stable Roommates}.
To this end, we exploit the parameters ``degree of change'' both in
the input (difference between old and new preference profile)
and in the output (difference between old and new stable matching).
We obtain both hardness and tractability results,
in particular showing a fixed-parameter tractability result with respect to the parameter ``distance between old and new stable matching''.
\end{abstract}

\section{Introduction}
Imagine the following scenario. A manager responsible for a group
of $2n$~workers has to form $n$ two-worker teams based on the preferences
over potential work partners of each worker.
The manager, being interested in a robust solution,
computes a stable matching (indeed, this refers
to the \textsc{Stable Roommates} problem). However,
say every month the workers may update their preferences
about wanted work partners
(the updates may be based on gained experiences, new information,
newly developed personal skills etc.). The team manager then has to find
a new stable matching
respecting the individually evolved preferences.
To this end, however, the team manager may not want to
allow too radical changes in the composition of the two-worker teams
because this might e.g.\ overburden administration.
Thus, a moderate change is acceptable, but too radical changes in
the team compositions should be avoided whenever possible.
We address this scenario
by introducing and studying ``incremental versions''
of the two most prominent stable matching problems, namely
\textsc{Stable Marriage} and \textsc{Stable Roommates}.\footnote{The
naming ``Incremental'' is inspired by work in the context of clustering and
information retrieval~\cite{ChChFeMo2004}.}

In stable matching scenarios or, in other words, in matching under preferences~\cite{Manlove2013}, one is
given a set of agents, each of them having preferences
over (some of) the other agents, and the goal is to match pairs of agents such
that the outcome is stable. Informally speaking, stability means that
there are no two agents that would both prefer to be matched
with each other instead to their current partners in the matching (or being unmatched).
Two classic problems here are the bipartite
case with two equal-sized sets of agents
(referred to as \textsc{Stable Marriage}) and
the general case of an even number of agents
(referred to as \textsc{Stable Roommates}).
Motivated by the introductory considerations, we next define
``incremental versions'' of both problems; 
further formal definitions are presented later.


\decprob{\ISM}
{Two disjoint sets~$U$ and~$W$ of~$n$ agents each, two preference profiles~$\mathcal{P}_1$ and~$\mathcal{P}_2$ for $U\uplus W$, a stable matching~$M_1$ for profile~$\mathcal{P}_1$, and a non-negative integer~$k$.}
{Does~$U \uplus W$ admit a stable matching~$M_2$ for profile~$\mathcal{P}_2$ such that~$\dist(M_1,M_2)=|M_1 \Delta M_2|\le k$?
}

Herein, $M_1 \Delta M_2$ denotes the symmetric difference
between sets $M_1$ and~$M_2$.
The incremental setting for \SR is defined analogously.

\decprob{\ISR}
{A set~$V$ of~$2n$ agents, two preference profiles~$\mathcal{P}_1$ and~$\mathcal{P}_2$ for $V$, a stable matching~$M_1$ for profile~$\mathcal{P}_1$, and a non-negative integer~$k$.}
{Does~$V$ admit a stable matching~$M_2$ for profile~$\mathcal{P}_2$ such that~$\dist(M_1,M_2)=|M_1 \Delta M_2| \le k$?
}

In both definitions
there are two main regulating screws. First, we are given \emph{two} preference profiles,
the old~$\mathcal{P}_1$ and the new~$\mathcal{P}_2$, thereby
reflecting the change of preferences. Indeed, to
reflect only moderate changes (``evolution''), we will subsequently measure the difference
between the two profiles (later referred to as swap distance), yielding a natural problem-specific parameter (the smaller it is,
the less revolutionary the changes are). Second,  the number~$k$ can be interpreted as
a locality parameter---it exposes how close the new matching has to be to the old one. The smaller we choose~$k$,
the more conservative we are with respect to change in the outcome (namely the difference between old and
new stable matching).
Together, we thus have one parameter to regulate the degree
of change measured in the \emph{input} preferences and
one parameter to regulate the degree
of change measured in the \emph{output} stable matching solution.
Taking up these two parameters,
we provide a thorough
parameterized complexity analysis of these kinds of stable matching problems with evolving preferences.\footnote{In the large
field of dynamic graph algorithms, which significantly differs from our
setting,
the term ``dynamic'' refers to more fine-grained scenarios where typically
edges and/or vertices may be added or deleted in a stepwise fashion,
and one wants to efficiently update a solution after every such single change.
In particular, this has also been studied in the (popular and stable) matching
context~\cite{BHHKW15,GKP17,NV19}. The main difference to our work is that
we study changes between two preference profiles (not so much in the graph structure), and
the changes can be at a larger scale. Moreover, we perform parameterized complexity studies
which do not play a role in this previous work.}
To this end,
we also distinguish between preferences
with and without ties. Roughly speaking, we provide positive (fixed-parameter) tractability results
in the case without ties and several (parameterized) intractability results for the
case with ties. Before describing our results in more detail, however,
we discuss related work.

\subsection{Related work}
There is previous work on matching-related problems in the context of dynamic
graph algorithms where vertices and/or edges arrive or depart
iteratively over time. The goal then is to maintain a solution of sufficient quality by performing necessary updates after every single change.
 The main difference to our work is that we study changes between two preference profiles (not so much in the graph structure), and the changes can be at a larger scale.
For instance, \citet{BHHKW15} studied the maintenance of
near-popular matchings (a scenario related to stable matchings)
based on a greedy improvement strategy, heavily employing maximum vertex degree
in their analysis. \citet{GKP17} and \citet{NV19} investigated dynamic rank-maximal and popular matchings,
again within the setting of dynamic graph algorithms and a focus on update times
after each single change in the graph.
\citet{KanLeoMag2016} studied \textsc{Stable Marriage} where at each time-step, two random adjacent agents in some preference list are swapped, and designed approximation algorithms to maintain a matching with logarithmic number of blocking pairs.

\citet{GenSiaSimOSul17cocoa,GenSiaOSuSim2017ijcai,GenSiaSimSul17aaai,genc_complexity_2019tcs} studied
\emph{robustness} of a matching in \textsc{Stable Marriage}. 
Herein, the robustness of a given stable matching is measured by the number of modifications needed to find an alternative stable matching if some currently
matched agent pairs break up.
They define an $(x, y)$-supermatch as a stable matching that satisfies the following: 
If any $x$~agents break up, then it is possible to rematch these $x$~agents so that the new matching is again stable; further, this re-matching \emph{must not} break up more than $y$~other pairs.
\citet{genc_complexity_2019tcs} showed that 
deciding whether a $(1,1)$-supermatch exists 
is \np-complete.
The main distinguishing features compared to our model are that,
using our two regulating screws mentioned above, we can model both moderate changes
in the preferences (that is, the input) and moderate differences between
the old and the new stable
matching (that is, the output). Indeed, we perform a parameterized complexity analysis exploiting
these parameters while \citet{GenSiaSimOSul17cocoa,GenSiaOSuSim2017ijcai,GenSiaSimSul17aaai,genc_complexity_2019tcs} focused on classic complexity results
and used heuristics in experimental work.
Moreover, we model changing preferences while
they addressed breaking up matched pairs.
Finally, our main contributions are in the \textsc{Stable Roommates} case while they exclusively
focused on \textsc{Stable Marriage} without ties.

Our model of measuring the distance between the input and sought matching is related to the stable matching problem with \emph{forbidden} and \emph{forced} edges studied by~\citet{CsehManlove2016}.
Th goal is to find a stable matching which minimizes the number of forbidden edges plus the number of non-forced edges.
While we do not model forbidden edges, \citet{CsehManlove2016} do not assume changes in agents' preferences. This makes a difference in terms of parameterized complexity.


\citet{MarxSchlotter2010,MarxSchlotter2011} studied
local search aspects for the \np-hard
\textsc{Stable Marriage} with ties.
More specifically, they investigated the parameterized complexity
of a local search variant of \textsc{Stable Marriage} using parameters
such as the number of ties. Thus, the main overlap with our work
is in terms of searching for local improvements and
employing parameterized complexity analysis---the studied
computational problems are different from ours as they do not model changes in the input preferences.

There have been numerous other models and investigations
to enrich the basic stable matching model, including
the use of only partially ordered preferences~\cite{DruBou13},
``multilayer'' stable matchings with several preference profiles
to be obeyed `in parallel'~\cite{AzBiGaHaMaRa2016,MiOk2017,CheNieSkoECmstable2018},
or studying robust stability in a probabilistic model~\cite{MaiVaz2018,MaiVaz2018-birkhoff-arxiv} or
from a quantitative angle~\cite{ChenSkowronSorge2019ec-robust}.

Finally, let us briefly mention that the motivation for our incremental scenario for stable matching
is related to similar
scenarios in the context of clustering~\cite{ChChFeMo2004,LuoMNN18},
coloring~\cite{HarNie2013}, other dynamic versions of parameterized
problems~\cite{Abu-Khzam+2015,krithika2018dynamic}, and
reoptimization~\cite{abs-1809-10578,schieber2018theory}.

\subsection{Our contributions}
Besides introducing a fresh model
of stable matching computations,
we provide results mostly in terms of parameterized complexity analysis
for both \textsc{Incremental Stable Marriage} and \textsc{Incremental Stable
Roommates}, where we see the main technical contributions
mostly for the latter.
In particular, our main algorithmic result is that
\textsc{Incremental Stable Roommates} for input instances without ties is fixed-parameter tractable with respect to the parameter~$k$ (distance between the old and the new
matchings). To show this, we heavily use structural results
due to \citet{Irving1994} and \citet{gusfield1988structure}
and show how to exploit them for designing
a fixed-parameter algorithm.

Most of our results are surveyed in Table~\ref{tab:results}. 
Herein, $\mathcal{P}_1 \oplus \mathcal{P}_2$ denotes the swap distance
between two preference profiles (see Section~\ref{sec:prelim} for formal
definitions).
The table indicates that we obtained a fairly complete picture of the computational
complexity landscape, e.g., also complementing \woneh{ness} results with corresponding
\xp-algorithms.

\begin{table}[tb]
\centering
\begin{tabular}{l l l@{\,}l c l@{\,}l}
 \toprule
Ties & $\left| \mathcal{P}_1 \oplus \mathcal{P}_2 \right|$ & \multicolumn{2}{c}{\textsc{Incr.} \SM} & & \multicolumn{2}{c}{\textsc{Incr.}  \SR} \\
 \midrule
no  & any &  P & (Prop.~\ref{prop:ISMviaWeightedSM})                       &      & \fpt{} for $k\coloneqq \dist(M_1,M_2)$ & (Thm.\ \ref{thm:ISR FPT without ties})   \\
 &  &                             & &  &\woneh{} for $k'\coloneqq |M_1 \cap M_2|$ & (Prop.\ \ref{prop:InSR no tie NP-hard})   \\[1ex]
yes &  1  &  \woneh{} for $k$ & (Thm.\ \ref{thm:ISMTiesIncompleteHardness}) & &\woneh{} for~$k$, even for compl. pref. & (Thm.\  \ref{thm:W1h-ISR-ties-ONESWAP}) \\ 
&&& &  &  \np-hard even if $k'\ge 0$ & (Thm.\ \ref{thm:W1h-ISR-ties-ONESWAP})  \\[1ex]
yes &  2  &  \woneh{} for $k'$ & (Thm.\ \ref{thm:ISMTiesDualParameter})    &    & \woneh{} for $k'$ & (Thm.\ \ref{thm:ISMTiesDualParameter}) \\[1ex]
yes & any &  \xp{} for $k$ & (Prop.~\ref{prop:XPalgorithmForK})                      &           & \xp{} for~$k$ & (Prop.~\ref{prop:XPalgorithmForK}) \\
\bottomrule
\end{tabular}
  \caption{\label{tab:results}
  Overview of our results.
  Unless otherwise stated, results are for the general case where preferences could be incomplete.
  }
\end{table}


\todo[inline]{JJ: XP for $k'$?}

Finally, we mention in passing that to achieve our results, throughout the work we
also introduce and study some intermediate
problems (for instance, \textsc{Edge-Incremental Independent Set}) which
may be of independent interest and prove useful in other settings.

\subparagraph*{Organization of the Paper.}
The rest of the paper is organized as follows.
We start with formal definitions and notations in \Cref{sec:prelim}.
Our main algorithmic result \cref{thm:ISR FPT without ties} and other algorithms are presented in \cref{sec:ISR tactable}.
Our hardness results are given in \Cref{sec:ISR instractable}.
We conclude in \Cref{sec:conclusion} with some open problems.

\section{Definitions and notations}
\label{sec:prelim}
In this section, we review fundamental concepts used in matchings under preferences.

\paragraph{Preference lists and profiles.}
Let $V=\{1,2,\ldots, 2n\}$ be a set of $2n$~agents.
Each agent~$i\in V$ has a subset~$V_i\subseteq V \setminus \{i\}$ of agents which they find \emph{acceptable} as partners and has a \emph{preference list~$\succeq_i$} on~$V_i$ ({i.e.}, a transitive and complete binary relation on $V_i$).
Here, $x \succeq_i y$ means that $i$ weakly prefers $x$ over $y$ ({i.e.}\ $x$ is at least as good as $y$).
We use $\succ_i$ to denote the asymmetric part ({i.e.}, $x\succeq_i y$ and $\neg (y\succeq_i x)$)
and $\sim_i$ to denote the symmetric part of $\succeq_i$ ({i.e.}, $x\succeq_i y$ and $y \succeq_i x$).
If $x \sim_i y$, then we also say that $x$ and $y$ are \emph{tied} in $i$'s preference list
and that agent~$i$ has \emph{ties} in their preference list.

For two agents~$x$ and $y$, we call $x$ \emph{most acceptable} to~$y$ if $x$ is a maximal element in the preference list of $y$. Note that an agent can have more than one most acceptable agent.
For two disjoint subsets of agents $X \subseteq V$ and $Y \subseteq V$, $X \cap Y = \emptyset$, we write $X\succeq_i Y$ if for each pair of agents~$x \in X$ and $y \in Y$ we have $x\succeq_i y$.

A preference profile~$\mathcal{P}$ for $V$ is a collection~$(\succeq_i)_{i\in V}$ of preference lists for each agent~$i\in V$.
A profile~$\mathcal{P}$ may have the following properties:
\begin{enumerate}
  \item
  It is \emph{complete} if for each agent~$i\in V$ it holds that $V_i \cup \{i\}= V$; otherwise it is \emph{incomplete}.
  \item
  The profile~$\mathcal{P}$ has \emph{ties} if there is an agent~$i \in V$ with ties in its preference list.
\end{enumerate}
To an instance $(V,\mathcal{P})$ we assign an \emph{acceptability graph}~$G$, which has $V$ as its vertex set and two agents are connected by an edge if each finds the other acceptable.
Without loss of generality, $G$~does not contain isolated vertices, meaning that each agent has at least one agent which it finds acceptable.

\paragraph{Swaps and differences between two matchings.}
Given two preference lists~$\succeq$ and $\succeq'$,
the \emph{swap distance between $\succeq$ and $\succeq'$} is defined as the number of pairs that are ``ordered'' differently; if $\succeq$ and $\succeq'$ are defined on different sets, then we assume that their swap distance is infinite.
\todo{Dušan: I do not think we ever use the distance for different acceptability -- can we simplify it to capture only the simple case?}
\todohua{Hua: I think if we only define different swap distances for the case with the same acceptability set then we need to make a disclaimer that we only consider restricted profiles. However, this disclaimer makes the paper appears less universal. So, I'd prefer not to do so.}
Formally,
\begin{align*}
  \sdist(\succeq, \succeq') \coloneqq
  \begin{cases}
    \infty, &  \succeq \text{ and } \succeq' \text{ are defined on } \text{\emph{different} sets,}\\
    \left|\left\{(x,y) \mid x\succ y \wedge y\succeq' x \right\}\right| + \\
    |\{(x,y) \mid x \sim y \wedge (x,y) \notin \sim'\}| , & \text{otherwise.}\\
  \end{cases}
\end{align*}
  For example, if an agent has a preference list on $\{a, b, c\}$ where all three agents are tied on the first position, that is, $a \sim b \sim c$,
then moving to the list $c \succ \{a, b\}$ requires two swaps.

Let $\mathcal{P}_1$ and $\mathcal{P}_2$ be two preference profiles for the same set~$V$ of agents.
The \emph{swap distance between $\mathcal{P}_1$ and $\mathcal{P}_2$}, denoted as $\left|\mathcal{P}_1 \oplus \mathcal{P}_2\right|$, is defined as the sum of the swap distances between the preference lists of each agent in the two preference profiles.
Formally, let $\succeq_i^j$ be the preference list of an agent $i \in V$ in the profile $\mathcal{P}_j$ for $j = 1,2$. We have that
$
  \left|\mathcal{P}_1 \oplus \mathcal{P}_2\right| = \sum_{i \in V} \sdist(\succeq_i^1, \succeq_i^2)
$.

Given a set $V$ of agents, a \emph{matching} $M$ of $V$ is a set of pairwisely disjoint pairs of agents in~$V$.
Given two matchings~$M_1$ and $M_2$ for the same set~$V$ with $2n$~agents, we define the \emph{distance between $M_1$ and $M_2$} as the size of the symmetric difference of $M_1$ and $M_2$, formally:
$  \dist(M_1, M_2) \coloneqq |M_1 \Delta M_2| = |M_1 \setminus M_2| +  |M_2 \setminus M_1|$.


\paragraph{Blocking pairs and stable matchings.}
Let a preference profile~$\mathcal{P}$ for a set~$V$ of agents, the corresponding acceptability graph~$G$, and a matching~$M\subseteq E(G)$ be given.
For a pair~$\{x,y\}$ of agents, if $\{x,y\}\in M$, then we denote the corresponding partner~$y$ by~$M(x)$; otherwise we call this pair \emph{unmatched}. We write $M(x)=\bot$ if agent~$x$ has \emph{no partner}, {i.e.}, if the agent~$x$ is not involved in any pair in~$M$.
We use $\bot(M)$ to denote the set of unmatched agents in a matching~$M$, that is, $\bot(M) = \{ x \mid M(x) = \bot\}$.
If \emph{no} agent~$x$ has $M(x)=\bot$~(i.e., $\bot(M) = \emptyset$), then $M$ is called \emph{perfect}.

Given a matching $M$ of $\mathcal{P}$, an unmatched pair~$\{x,y\}\in E(G)\setminus M$ \emph{is blocking} $M$ if both~$x$ and $y$ prefer each other to being unmatched or to their assigned partners, {i.e.}\ it holds that $\big(M(x)=\bot \vee y\succ_x M(x) \big) \wedge \big( M(y)=\bot \vee x \succ_y M(y) \big)$.
We call a matching~$M$ \emph{stable}\footnote{We exclusively focus on the weak stability concept~\cite{Manlove2013}. We conjecture that several results will also hold
at least for strong stability. The proofs, however, may need non-trivial adjustments.}
if no unmatched pair is blocking $M$.
The \SR problem is defined as follows:
\decprob{Stable Roommates (SR)}
{A preference profile~$\mathcal{P}$ for a set~$V=\{1,2,\ldots, 2n\}$ of $2n$ agents.}
{Does $\mathcal{P}$ admit a stable matching?}

\paragraph{{\normalfont \textsc{Stable Roommates}} and {\normalfont \textsc{Stable Marriage}}.}
The bipartite variant of \textsc{Stable Roommates}, called \textsc{Stable Marriage}, has as input two $n$-element disjoint sets~$U \uplus W$ of agents such that each agent~$u \in U$ from~$U$ has a \emph{preference list on $W_u \subseteq W$}  and each agent~$w\in W$ from~$W$ has a \emph{preference list on $U_w \subseteq U$}.
By definition, the underlying acceptability graph of a \textsc{Stable Marriage} instance is bipartite.
Accordingly, this instance has complete preferences if this graph is a complete bipartite graph.
%



The following fundamental result from the literature guarantees that we can deal with incomplete preferences without ties similarly to the case with complete preferences with ties.

\begin{proposition}[{\cite[Theorem~1.4.2, Theorem~4.5.2]{GusfieldIrving1989}}]\label{prop:IncompleteNoTies}
   For incomplete preferences without ties, the whole agent set can be partitioned into two disjoint subsets~$R$ and $S$ such that every stable matching matches every agent from~$R$ and none of the agents from~$S$.
  For agent sets of size~$2n$, this partition can be computed in $O(n^2)$~time.
\end{proposition}

\section{Algorithms for Incremental Stable Marriage and Roommates} 
\label{sec:ISR tactable}

As a warm-up, we first show that our most restricted problem variant, \ISM without ties,
can be solved in polynomial time.
The main idea behind this result is based on the fact that there exists an compact and polynomial-time computable representation of all stable matchings, the so-called partially ordered set of polynomially many rotations~\cite[Chapter 2]{GusfieldIrving1989}.
Herein, a rotation involves a subset of some stable matching ``reducing'' which results in another stable matching~\cite[Chapter 2.5.1]{GusfieldIrving1989}.
Equipping these rotation with some appropriate weight, we can reduce our problem to finding a closed subset of rotations with maximum weight, which can be solved in polynomial-time
using an approach similar to the one discussed in \cite[Chapter 3.6.1]{GusfieldIrving1989}.
We refer to \cref{sec:ISM without ties} for proof details.

\begin{proposition}\label{prop:ISMviaWeightedSM}
\ISM without ties can be solved in~$O(n^3)$ time.
\end{proposition}
\appendixproof{prop:ISMviaWeightedSM}{
\begin{proof}
  The result from \cref{prop:ISMviaWeightedSM} follows from \cref{app:lem:sd-inter} and \cref{app:lem:SM-max-w-rotations}.
\end{proof}
}

Second we show that our least restricted \np-hard problem variant, \ISR with ties allowed,
can at least be solved in polynomial time when the distance between the matchings~$M_1$
and~$M_2$ is a constant. In other words, \cref{prop:XPalgorithmForK} presents
an \xp-algorithm for the distance as a parameter.

\begin{proposition}\label{prop:XPalgorithmForK}
  \ISR with ties can be solved in $n^{O(k)}$~time, where $k = \left| M_1 \Delta M_2 \right|$.
\end{proposition}
\begin{proof}
  Let $\hat{n} = |M_1|$.
  The algorithm first guesses positive integers $k_1,k_2$ with $k_1 + k_2 \le k$ which we further treat as follows.
  We denote by $k_1$ the number of matching edges leaving~$M_1$ (that is, $|M_1 \setminus M_2|$) and by $k_2$ the number of new edges (that is, $|M_2 \setminus M_1|$).
  Now, for each such pair $(k_1,k_2)$ we first guess $k_1$ edges to delete from~$M_1$.
  Notice that there are $O(\hat{n}^{k_1}) = O(n^k)$ possible guesses.
  Let $\hat{M}_1$ be the rest of the matching~$M_1$ (after we delete the just guessed edges).
  Then we guess $k_2$ pairs of vertices not matched in~$\hat{M}_1$.
  This completely defines the matching~$M_2$ which can be checked in polynomial-time for stability.
  Notice that there are $\binom{2n - 2\hat{n} + 2k_1}{2k_2} \cdot \frac{(2k_2)!}{2 \cdot k_2!} = O(2k^{2k} \cdot (2n)^{2k})$ possible guesses.
  Repeating the above procedure for every pair $(k_1,k_2)$ and applying the obvious bound $k \le n$, we get the overall running time $n^{O(k)}$ as claimed.
\end{proof}

We remark that \cref{prop:XPalgorithmForK} can also be shown by using an idea of \todo{Add reference.}
Our approach, however, is much simpler.
The most pressing question following from \cref{prop:XPalgorithmForK} is to ask for which cases
this result can be improved to a fixed-parameter tractability result for the parameter~$k$, the distance between~$M_1$ and~$M_2$.
Unfortunately, we will show in the next section that this is not possible when ties are involved.
In the remainder of this section, however, we show that this is possible when ties are not allowed.
%
%
%
Formally, we show the following.
\begin{theorem}
\label{thm:ISR FPT without ties}
\ISR without ties can be solved in~$O(2^k n^4)$ time.
\end{theorem}
%
The algorithm behind \cref{thm:ISR FPT without ties} is partially inspired by the polynomial-time algorithm for solving \textsc{Maximum Weight Stable Marriage}~\cite[Chapter 2.5.1]{GusfieldIrving1989}.(also see \cref{sec:ISM without ties} for more details). 
The high-level structure of our algorithm is as follows: 
\begin{enumerate}
  \item Using the algorithm of Irving~\cite{irving1985efficient} and based on the structural insights of Gusfield~\cite{gusfield1988structure} that the set of all stable matchings can be compactly and efficiently represented via the so-called partially ordered set (poset) of polynomially many rotations, we compute the poset of rotations for the new preference profile $P_2$  (see \cref{sec:rotationpreliminaries}).
  
  Roughly speaking, a rotation (relative to a preference profile) involves a specific cyclic sequence~$\sigma$ of agents where the first acceptable partner in the preference list of each agent $x$ in $\sigma$ is exactly the second acceptable partner in the preference list of $x$'s predecessor in~$\sigma$. We remark that for a rotation, removing the first acceptable partner of each agent~$\sigma$ will not alter the existence of a stable matching.

 \item To capture the (different) costs of eliminating one rotation from each dual pair
 (see \cref{def:dual-single-pairs})
  with respect to the resulting
       distance between the matching~$M_1$ and $M_2$, we compute two weights for each rotation (see \cref{subsec:proposalsets+weights}).
 \item Finally, to consistently choose one rotation from each dual pair, and additionally respecting the weight constraints,
       we reduce our problem to an auxiliary problem to be defined below and provide a fixed-parameter algorithm for this auxiliary problem (see \cref{subsec:reduction-to-WCFCS}).
\end{enumerate}
Each high-level step of the algorithm will be described in one of the subsequent subsections.

To define our auxiliary problem (which may be of independent interest), for each partially order set~$(R, \predr)$,
we say that a given subset~$C \subseteq R$ is \emph{closed with respect to} the relation $\predr$ if
for each pair~$\{\rho,\sigma\} \subseteq R$ of elements it holds that $\rho \in C$ and $\sigma \predr \rho$ imply $\sigma \in C$.
\decprob{\WCFCS}
{A~partially ordered set~$(R, \predr)$, an undirected graph $G = (R, E)$, a weight function $w\colon R \to \mathds{N}$, a positive integer~$\ell$, and a budget~$b \in \mathds{N}$.}
{Is there a closed subset~$C \subseteq R$ of size~$\ell$ which is independent in~$G$ such that \mbox{$\sum_{c \in C} w(c) \le b$?}}
Here, the conflicts are modeled by the edges so that the graph is also referred to as a \emph{conflict graph} and seeking for an independent set solution means seeking for a \emph{conflict-free set}.
\todohua{Hua: GusfieldIrving1989 used the name ``maximal independent set'' to refer to the solution. So I would tend to use the notion already introduced there if we wanted to introduce this additional problem.}
\todo{Dušan: I did not get to what notion does ``maximal independent set'' apply.}
\todohua{Hua: ``Independent set'' comes from the problem requirement, maximality comes from selecting exactly one rotation from each dual pair.}
\todohua{Hua: I don't get the meaning of ``conflict-free'', which does not even appear in the problem definitions (except the problem name).}
The structure of the rotation poset implies that in our case the graph $G$ in fact consists of $\ell$ edges forming a matching of size $\ell$.
In particular, we will show that we can focus on the special case in which the conflict\todohua{Hua: Why conflict?  I thought that we are searching for an independent set...} graph $G$ consists of $\ell$~disjoint cliques.
For this case we give an algorithm for \WCFCS for parameter $b+\Delta(G)$, where $\Delta(G)$ is the maximum vertex degree in the graph $G$.
Moreover, it follows from our reduction procedure that $b \le k$ (see \cref{subsec:reduction-to-WCFCS}).

\subsection{Preprocessing and identification of the rotations}
\label{sec:rotationpreliminaries}
We first recall Irving's polynomial-time algorithm~\cite{irving1985efficient} for determining whether an instance of \SR without ties has a stable matching and fundamental structural properties behind all stable matchings~\cite{gusfield1988structure}.
Irving's algorithm is divided into two phases: Phase~1 and Phase~2.
Phase~1 involves a sequence of proposals from each agent~$i$ to the first agent~$j$ on $i$'s list, and each such proposal resulting in the deletion of all successors of~$i$ from~$j$'s list.
Phase~1 does not alter any stable matching since in this phase~$j$ is removed from~$i$'s preference list (and~$i$ is removed from~$j$'s preference list) only when~$\{i,j\}$ does not form a pair in any stable matching.
The set of lists at the end of the first phase is called \emph{Phase~1 table} (please refer to \cref{example:tablesRotationsElimination} for an illustration).

There are three possibilities for the Phase~1 table; note that the Phase~1 table is unique.
\begin{description}
  \item[Every list in the table is empty.] Then, there is no stable matching; note that Irving's original algorithm assumes that the input preferences are complete, implying that whenever a list becomes empty, then there is not stable matching; for incomplete preferences we can only infer the non-existence of any stable matching if every list becomes empty.
\item[Every non-empty list contains exactly one agent.] Then, we find a unique stable matching.
\item[At least one list contains more than one agent.] Then, we proceed to Phase~2. 
\end{description}

For agents with empty preference lists in the Phase~1 table, we use the following.
\begin{proposition}[{\cite[Theorem~4.5.2]{GusfieldIrving1989}}]\label{prop:empty-list-Phase1-unmatched}
  For each agent, if its preference list becomes empty after Phase~1,
  then it will not be matched by any stable matching; otherwise it must be matched by all stable matchings.
\end{proposition}

Since in the first two cases the instance of \ISR is trivial to solve, we assume the third case in the following and we can ignore every agent whose list becomes empty after Phase~1.
Notice that in the Phase~1 table we always have that if an agent $i$ is ranked first in $j$'s preference list, then $j$ is ranked the last in $i$'s preference list.
This invariant we keep throughout the whole Phase~2.
We need some definitions for Phase~2.
A \emph{preference table} is the Phase~1 table or any Phase~2 table (i.e., it is a collection of preference lists).

\begin{definition}[Rotations]
A rotation \emph{exposed} in a preference table~$\mathcal{T}$ is an ordered sequence \[(e_0,h_0),(e_1,h_1),\dots,(e_{r-1},h_{r-1})\] of pairs of agents such that~$h_i$ and $h_{i+1 \bmod r}$ are the first and the second agent on~$e_i$'s list in $\mathcal{T}$, for all~${0 \le i \le r-1}$.
\end{definition}
Note that if $(e,h)$ is a pair in a rotation, then $e$ is ranked last by $h$ with respect to $\mathcal{T}$.

\begin{definition}[Elimination of an exposed rotation]
The \emph{elimination} of an exposed rotation~$(e_0,h_0)$, $(e_1,h_1),\dots,(e_{r-1},h_{r-1})$ from table~$\mathcal{T}$ is the following operation.
For every~$0 \le i \le r-1$, remove every entry below~$e_i$ in~$h_{i+1}$'s list in~$\mathcal{T}$, i.e., move the bottom of~$h_{i+1}$'s list up to~$e_i$ (from~$e_{i+1}$).
Then for each agent~$p$ who was just removed from~$h_{i+1}$'s preference list, remove~$h_{i+1}$ from~$p$'s~list.
\end{definition}

In Phase~2, exposed rotations (see \cref{example:tablesRotationsElimination}) are eliminated from the preference table one by one until some list becomes empty, which means there is no stable matching, or no rotation is exposed in the table, which means the list of every agent contains exactly one agent and we get a stable matching.

\begin{proposition}[{\cite[Corollary~3.2]{Irving1994}}]\label{prop:emptylist-Phase2->nostablematching}
  If some agent obtains an empty preference list in Phase~2, then there is no stable matching.
\end{proposition}

To get better acquainted with Irving's algorithm and, more importantly, the notion of rotations, their exposition, and their elimination, we give an example illustrating this.
\begin{example}[{\cite{gusfield1988structure}}]\label{example:tablesRotationsElimination}
Consider the following profile with eight agents---the complete profile is to the left and the Phase~1 table is to the right.
For the sake of readability, the symbol~$\succ$ between the agents in the preference lists of each agent are omitted.

\begin{multicols}{2}
\begin{tabular}{l|lllllll}
    1 & 7 & 2 & 6 & 8 & 5 & 3 & 4 \\
    2 & 4 & 6 & 5 & 3 & 8 & 1 & 7 \\
    3 & 5 & 2 & 1 & 7 & 4 & 6 & 8 \\
    4 & 1 & 7 & 3 & 6 & 5 & 8 & 2 \\
    5 & 7 & 1 & 8 & 4 & 6 & 2 & 3 \\
    6 & 7 & 3 & 8 & 4 & 5 & 1 & 2 \\
    7 & 2 & 8 & 4 & 3 & 5 & 6 & 1 \\
    8 & 4 & 2 & 3 & 5 & 6 & 7 & 1 \\
\end{tabular}
\qquad\qquad
\begin{tabular}{l|lllllll}
    1 & 2 & 6 & 5 & 3 & 4 \\
    2 & 6 & 5 & 3 & 8 & 1 \\
    3 & 5 & 2 & 1 & 7 & 4 & 6 \\
    4 & 1 & 7 & 3 & 6 & 5 & 8 \\
    5 & 7 & 1 & 8 & 4 & 6 & 2 & 3 \\
    6 & 3 & 8 & 4 & 5 & 1 & 2 \\
    7 & 8 & 4 & 3 & 5 \\
    8 & 4 & 2 & 5 & 6 & 7 \\
\end{tabular}
\end{multicols}

In the Phase~1 table, there are two rotations: $\rho_1=((1,2),(2,6),(3,5))$ and~$\rho_2=((4,1),(5,7))$.
Now, we eliminate the rotation~$\rho_1$, that is, we remove
\begin{itemize}
  \item the top choices for agents $1,2,3$ and
  \item $8,1$ from the preference list of agent $2$, since $2$ finds $3$ better than these two and is becoming the top choice for $3$.
\end{itemize}
After the elimination of~$\rho_1$, we obtain the following table.

\begin{tabular}{l|lllllll}
    1 & 6 & 5 & 3 & 4 \\
    2 & 5 & 3 \\
    3 & 2 & 1 & 7 & 4 & 6 \\
    4 & 1 & 7 & 3 & 6 & 5 & 8 \\
    5 & 7 & 1 & 8 & 4 & 6 & 2 \\
    6 & 3 & 8 & 4 & 5 & 1 \\
    7 & 8 & 4 & 3 & 5 \\
    8 & 4 & 5 & 6 & 7 \\
\end{tabular}

Now $\rho_2$ is still exposed, and there is a new exposed rotation~$\rho_3=((2,5),(6,3),(7,8),(8,4))$.
If we continue to eliminate~$\rho_2$, then we obtain the folllwing table.

\begin{tabular}{l|lllllll}
    1 & 6 & 5 \\
    2 & 5 & 3 \\
    3 & 2 & 4 & 6 \\
    4 & 7 & 3 & 6 & 5 & 8 \\
    5 & 1 & 8 & 4 & 6 & 2 \\
    6 & 3 & 8 & 4 & 5 & 1 \\
    7 & 8 & 4 \\
    8 & 4 & 5 & 6 & 7 \\
\end{tabular}

Now $\rho_3$ is the only exposed rotation. We eliminate $\rho_3$ and obtain the following table.

\begin{tabular}{l|lllllll}
    1 & 6 & 5 \\
    2 & 3 \\
    3 & 2 \\
    4 & 7 \\
    5 & 1 & 8 \\
    6 & 8 & 1 \\
    7 & 4 \\
    8 & 5 & 6 \\
\end{tabular}

Now $\rho_4=((1,6),(8,5))$ and $\rho_5=((5,1),(6,8))$ are exposed.
If we continue to eliminate~$\rho_4$, we obtain the following table, where every person has exactly one entry on his list, and we get a stable matching.

\begin{tabular}{l|lllllll}
    1 & 5 \\
    2 & 3 \\
    3 & 2 \\
    4 & 7 \\
    5 & 1 \\
    6 & 8 \\
    7 & 4 \\
    8 & 6 \\
\end{tabular}
\end{example}


There are two types of rotations, which can be used to characterize all stable matchings.
Let~$D$ denote the execution tree when Phase~2 is executed in all possible ways on the Phase~1 table.
Note that each node~$x$ in~$D$ represents a table~$\mathcal{T}(x)$, which is the current table state of the algorithm at node~$x$.
Let~$R$ be the set of all rotations which are exposed in some table of some node in~$D$.

\begin{definition}[Dual and singleton rotations]\label{def:dual-single-pairs}
If~$\rho \in R$ with $\rho = ((e_0,h_0),(e_1,h_1),\dots,(e_{r-1},h_{r-1}))$, then the \emph{negation of~$\rho$} is defined as
  $\neg \rho=((h_0,e_{r-1}),(h_1,e_0),\dots,(h_{r-1},e_{r-2}))$.

If~$\neg \rho \in R$, then we call~$\rho$ and~$\neg \rho$ a \emph{dual pair of rotations}.
Any rotation without a dual is called a \emph{singleton rotation}.\footnote{The singleton rotation is called \emph{singular rotation} in \citet{GusfieldIrving1989,Manlove2013}; however, since we closely follow \citet{gusfield1988structure}, we prefer to use the term singleton rotation. This is mainly to make it easier for the reader when referring to lemmata of Gusfield.}
\end{definition}

\begin{definition}[Partially ordered set, and closed and complete subest of rotations]
  A subset~$C\subseteq R$ of rotations is \emph{complete} if it contains
  \begin{enumerate}[(i)]
    \item all singleton rotations and
    \item exactly one rotation from each dual pair of rotation.
  \end{enumerate}  

  A binary relation~$\pred$ on $R$ is defined as follows.
  For each pair~$\{\rho_1, \rho_2\}$ of rotations, if the elimination of~$\rho_1$ is necessary  for~$\rho_2$ to be exposed (i.e., the elimination of $\rho_1$ precedes the exposition of $\rho_2$ on every path in $D$),
  then we say that~$\rho_1$ precedes~$\rho_2$, written as \emph{$\rho_1 \pred \rho_2$}.
  The partial order~$\predr$ is the reflexive closure of the precedence relation~$\pred$.
  A subset~$C$ of the rotation~$R$ is \emph{closed with respect to the partial order~$\predr$} (in short, \emph{closed})  if for each pair~$\{\rho_1, \rho_2\}$ of rotations it holds that if $\rho_2 \in C$ and~$\rho_1 \predr \rho_2$, then~$\rho_1 \in C$.
\end{definition}


Note that the rotation poset~$\predr$ on~$R$ is an $O(n^2)$-sized representation of all stable matchings.
It follows from~Gusfield~\cite{gusfield1988structure} that one can compute~$R$ in~$O(n^3\cdot \log(n))$ time.


\begin{lemma}[{\cite{gusfield1988structure}}]
\label{lem:correspondence+time}
For a given \SR{} instance without ties, there is a one-to-one correspondence between the set of stable matchings and complete and closed subsets of the rotation poset~$(R, \rhd)$.
Finding this poset~$(R, \predr)$ can be done in $O(n^3\cdot \log(n))$~time, where $n$ denotes the number of agents in the instance.
\end{lemma}
According to \cref{lem:correspondence+time}, to solve our problem, it suffices to search for a complete and closed subset of rotations for the second input profile~$\mathcal{P}_2$ such that the corresponding stable matching is closest to~$M_1$.

\subsection{Proposal sets and rotation weights}\label{subsec:proposalsets+weights}
In this section, we study the influence of a rotation elimination on the distance between the ultimate resulting matching and the initial matching.
In order to do so, we first define proposal sets for preference tables and weights of rotations.
Here, the weight of a rotation shall capture how many pairs of the initial matching we can additionally obtain if we eliminate this rotation.
Then we study the properties of the thus introduced rotation weights.
This is used later to finish the sought reduction and to determine upper bounds on parameter values of the resulting instance.

Imagine that in a preference table, every agent proposes to the first agent in their current preference list.
Then, we get a set of ordered pairs, where the agent indicated by the first component of each pair proposes to the agent indicated by the second component in the pair.
\begin{definition}[Proposal sets]
For each preference table~$\mathcal{T}$, 
the \emph{proposal set} for~$\mathcal{T}$ is defined as
\[
\proposalset_\mathcal{T}\coloneqq \{(i,j) \mid i \in U \text{ and } j \text{ is the first agent on } i\text{'s list in }\mathcal{T} \} \,,
\]
where~$(i,j)$ represents a proposal pair~$i \rightarrow j$ (i.e., agent $i$ proposes to agent $j$). 

For each matching~$M$, the \emph{proposal set~$\proposalset_M$} for $M$ is defined
as $\proposalset_M\coloneqq \{(i,j),(j,i) \mid \{i,j\} \in M\}$.
\end{definition}

By the definition of proposal sets, for each matching~$M$ we have
\begin{equation}\label{eq:distanceByProposePairs}
  \dist(M_1,M) = |M_1\Delta M| = |M_1| + |M| - 2|M_1 \cap M| = |M_1| + |M| - \left| \proposalset_{M_1} \cap \proposalset_{M}\right| \,.
\end{equation}
Hence, we are looking for a matching $M_2$ which is stable with respect to profile $\mathcal{P}_2$ such that
\begin{equation}\label{eq:intersectionByProposePairs}
  \left| \proposalset_{M_1} \cap \proposalset_{M_2} \right| \ge |M_1| + |M_2| - k.
\end{equation}

Since every complete and closed rotation subset contains all singleton rotations and since all singleton rotations can be eliminated from the Phase~1 table before all dual rotations~\cite{gusfield1988structure}, we can first eliminate all singleton rotations.
By \cref{prop:emptylist-Phase2->nostablematching}, if after eliminating all singleton rotations, some agent obtains an empty preference list, then we can immediately conclude that the given profile does not admit any stable matching.
Thus, in the following, we mainly work with~$R_2 \subseteq R$ which is the set of all dual rotations.

To measure the benefit of eliminating a rotation, we define the following.
\begin{definition}[Gain and lost proposal pairs after a rotation elimination]\label{def:gain-lost}
For each rotation~$\rho \in R$ with~$\rho=(e_0,h_0),(e_1,h_1),\dots,(e_{r-1},h_{r-1})$,
\emph{the set of proposal pairs gained (lost) by eliminating rotation $\rho$} is defined as follows.
\[
\proposalset_\rho^+\coloneqq \{(e_0,h_1), (e_1,h_2), \dots, (e_{r-1},h_0)\},
\]
\[
\proposalset_\rho^- \coloneqq \{(e_0,h_0), (e_1,h_1), \dots, (e_{r-1},h_{r-1})\} \,.
\]

Further, let~$w^+(\rho)\coloneqq |\proposalset_\rho^+ \cap \proposalset_{M_1}|$ and~$w^-(\rho) \coloneqq |\proposalset_\rho^- \cap \proposalset_{M_1}|$ be the number of proposal pairs gained and lost by the elimination of rotation $\rho$.
\end{definition}
Observe that the two sets~$\proposalset_\rho^{+}$ and $\proposalset_\rho^{-}$ are independent of the table $\mathcal{T}$ and all of the proposal pairs of agents not involved in the rotation $\rho$ remain the same before and after the elimination of~$\rho$.

Since in order to obtain a stable matching, we have to eliminate exactly one rotation from each dual pair of rotations, we now prove that their weights are complementary.

\begin{lemma}
\label{lem:w(p)=-w(-p)}
Let~$\rho \in R$ be a dual rotation, then $w^+(\rho)=w^-(\neg \rho)$. 
\end{lemma}
\appendixproof{lem:w(p)=-w(-p)}{
\begin{proof}
  Let~$\rho\coloneqq (e_0,h_0),(e_1,h_1),\ldots,(e_{r-1},h_{r-1})$,
  and let~$\proposalset_\rho^+$ and~$\proposalset_\rho^-$ be the sets of proposal pairs gained and lost by eliminating~$p$, respectively.
That is,
\begin{align}
\proposalset_\rho^+= &\{(e_0,{h_1}), (e_1,{h_2}), \dots, (e_{r-1},{h_0})\}, \text{ and } \label{eq:proposal-rho+} \\
\proposalset_\rho^-= & \{(e_0,{h_0}), (e_1,{h_1}), \dots, (e_{r-1},{h_{r-1}})\}.\label{eq:proposal-rho-}
\end{align}
By the definition of dual rotations, the dual of $\rho$ is equal to $\neg \rho=(h_0,e_{r-1}),(h_1,e_0),\dots,(h_{r-1},e_{r-2})$.
Accordingly, we have:
\begin{align}
\proposalset_{\neg \rho}^+= & \{(h_0,{e_0}), (h_1,{e_1}), \dots, (h_{r-1},{e_{r-1}})\}, \text{ and } \label{eq:proposal-neg-rho+}\\
\proposalset_{\neg \rho}^-= & \{(h_0,{e_{r-1}}), (h_1,{e_0}), \dots, (h_{r-1},{e_{r-2}})\}.  \label{eq:proposal-neg-rho-}
\end{align}
By \eqref{eq:proposal-rho+} and \eqref{eq:proposal-neg-rho-}, we can infer that 
\begin{equation}
 \text{for each ordered pair~} (i,j)\text{ of agents it holds that } (i,j) \in \proposalset_\rho^+ \textrm{ if and only if } (j,i) \in \proposalset_{\neg \rho}^-. \label{eq:matchingsForNegations}
\end{equation}
Since~$M_1$ is a matching,
it follows that for each unordered pair~$\{i,j\}\in M_1$ of agents we have that $(i,j),(j,i) \in \proposalset_{M_1}$.
Together with~\eqref{eq:matchingsForNegations}, we conclude that
\begin{align*}
  \text{  for each ordered pair } (i,j)
  \text{ of agents it holds that } (i,j) \in \proposalset_\rho^+ \cap \proposalset_{M_1}
  \text{ iff. } (j,i) \in \proposalset_{\neg \rho}^- \cap \proposalset_{M_1}.
\end{align*}
Thus, 
\begin{equation*}
w^+(\rho)=|\proposalset_\rho^+ \cap \proposalset_{M_1}|=|\proposalset_{\neg p}^- \cap \proposalset_{M_1}|=w^-(\neg \rho).
\end{equation*}
\end{proof}
}

\subsection{Reduction to \WCFCS}\label{subsec:reduction-to-WCFCS}
In this subsection we give a reduction from \ISR to \WCFCS in which the conflict graph is 
is a union of disjoint edges.
In order to do so, we show that the distance between the target (i.e., initial) matching~$M_1$ and a matching~$M_C$ resulting from the elimination of a complete and closed set~$C$ of rotations for~$\mathcal{P}_2$ (if such set exists) depends only on~$\sum_{\rho \in C} w^-(\rho)$.

In the remainder of the section, let $\proposalset_0$ be the proposal set for the table obtained from the Phase~1 table for $\mathcal{P}_2$ followed by elimination of all singleton rotations to it and let
$R_2$ be the set of all dual rotations for $\mathcal{P}_2$; it is known that $R_2$ has at most ${n \choose 2}$ rotations.
\begin{lemma}[{\cite[Corollary~5.1]{gusfield1988structure}}]\label{lem:bound-rotation-set}
 $|R_2| \le {n \choose 2}$.
\end{lemma}

\paragraph{Upper-bounding the sum of weights of a complete and closed subset.}
Before we continue with the procedure, we compare the sizes of $M_C$ and $M_1$; recall that $M_c$ is a matching resulting from eliminating a complete and closed set~$C \subseteq R$.
By \cref{prop:empty-list-Phase1-unmatched}, every agent that has non-empty list after Phase~1 must be matched under $M_C$ and these agents are exactly those agents who hold some proposals in $\proposalset_0$.
Consequently, for each agent~$x$ that is matched under $M_1$ but does not hold a proposal under $\proposalset_0$ it holds that $\{x, M_1(x)\} \in M_1 \Delta M_C$.
Thus, we define the following reduction rule. 

\begin{rrule}\label{rr:not-matched-in-M0}
  For each agent~$x$ matched under $M_1$,
  if for each agent~$y$ it holds that  $(x,y) \notin \proposalset_0$, then delete $\{x,M_1(x)\}$ from $M_1$ and decrease $k$ by one.
\end{rrule}

\begin{lemma}\label{lem:rr-not-matched-in-M0-sound}
  \cref{rr:not-matched-in-M0} is sound and can be implemented in $O(n^2)$~time.
  Moreover, if \cref{rr:not-matched-in-M0} does not apply, then for each stable matching~$M$ of profile~$\mathcal{P}_2$ it holds that
  $|M_1| \le |M| = |\proposalset_0|/2$.
\end{lemma}
\appendixproof{lem:rr-not-matched-in-M0-sound}{
  \begin{proof}
    Let $x$ be as defined in the rule, i.e., $x$ is matched under $M_1$ but does not hold a proposal in $\proposalset_0$.
    For ease of notation, let $M^-_1 \coloneqq M_1 \setminus \{\{x,M_1(x)\}\}$.
    Let $M$ be a stable matching for $\mathcal{P}_2$.
    By \cref{prop:empty-list-Phase1-unmatched}, it follows that 
    $x$ is unmatched under $M$.
    Thus, $M_1 \Delta M = (M \Delta M^-_1) \cup \{x, M_1(x)\}$.
    This implies that $|M_1 \Delta M| \le k$ if and only if $|M^-_1 \Delta M| \le k-1$.
    The soundness of the reduction rule follows.

    As for the running time, for $2n$ agents in $\mathcal{P}_2$, in $O(n^2)$ time we can complete Phase~1 and obtain the proposal set~$\proposalset_0$ which has size at most $2n$.
    Then, in $O(n)$~time we can check for each matched agent (under $M_1$) whether it is ``unmatched'' in $\proposalset_0$.

    As for the second statement, it is immediate that $|M_1|\le |M|$.
    Again, by \cref{prop:empty-list-Phase1-unmatched}, it holds that $|M|=|\proposalset_0|/2$. 
    Hence, $|M_1|\le |M| = |\proposalset_0|/2$.
  \end{proof}
}

Next, we upper-bound the size of the intersection between the target (i.e., initial) matching and the sought matching.\todo{JJ: Upper-bound the sum of weights of a complete and closed subset? Hua will check.}
\begin{lemma}\label{lem:equivalence}
Let~$C$ be a complete and closed subset of rotations in $R_2$ and let~$M_C$ be the stable matching associated with~$C$. 
Then, the following holds.
\begin{enumerate}[(i)]
  \item  $|\proposalset_{M_1} \cap \proposalset_{M_C}| = |\proposalset_{M_1} \cap \proposalset_0| + \sum_{\sigma\in R_2}w^{+}(\sigma) - 2\cdot \sum_{\rho \in C} w^{-}(\rho)$.
  \item\label{lem:weight-bound-equivalence} $\dist(M_1, M_C) \le k$ if and only if  $\sum_{\rho \in C} w^{-}(\rho) \le \frac{|\proposalset_{M_1}\cap \proposalset_0| + \sum_{\sigma \in R_2}w^{+}(\sigma) - |M_1| - |M_C| +  k}{2}$.
  \item\label{lem:weight-bound} $|\proposalset_{M_1}\cap \proposalset_0| + \sum_{\sigma \in R_2}w^{+}(\sigma) - |M_1| - |M_C| \le 0$.
\end{enumerate}
\end{lemma}
\appendixproof{lem:equivalence}
{
  \begin{proof}
    To show the first statement, let us first observe the following,
    which is similar to the case in the marriage setting (see \cref{app:lem:SM-rotation-weight} and \cref{app:cor:SM-rotation-weights}).
    \begin{claim}\label{claim:SR-rotation-weight}
      Let $\mathcal{T}$ be a preference table and $\rho$ be a rotation exposed in $\mathcal{T}$,
      and let $\mathcal{T}'$ be the preference table obtained from $\mathcal{T}$ by eliminating $\rho$.
      Then, $|\proposalset_{M_1} \cap \proposalset_{\mathcal{T}'}| = |\proposalset_{M_1} \cap \proposalset_{\mathcal{T}}| + w^{+}(\rho) - w^{-}(\rho)$.
    \end{claim}
    \begin{proof}\renewcommand{\qedsymbol}{(of \cref{claim:SR-rotation-weight})~$\diamond$}
      The proof is similar to the one given for \cref{app:lem:SM-rotation-weight}.
      Let $\rho \coloneqq ((e_0, h_0), \cdots, (e_{r-1}, h_{r-1}))$ be a rotation exposed in the table~$\mathcal{T}$.
      In the following, all subscripts~$i+1$ are taken modulo~$r$.
      By the definition of $\mathcal{T}'$ and the definition of proposal sets we have that
      \begin{align}
        \nonumber        \proposalset_{\mathcal{T}'} = &
                               \{(x,y) \in \proposalset_{\mathcal{T}} \mid
                               (x,y) \notin \rho\} \uplus \{(e_i, h_{i+1})\mid 0 \le i \le r-1\}\\
        = & (\proposalset_{\mathcal{T}} \cup \proposalset_{\rho}^{+}) \setminus \proposalset_{\rho}^{-}.\label{eq:Table-Dif}
        \end{align}
      Thus, we prove the statement in the claim by showing the following. 
  \begin{align*}
    |\proposalset_{M_1}\cap \proposalset_{\mathcal{T}'}|
    \stackrel{\eqref{eq:Table-Dif}}{=} & |\proposalset_{M_1}\cap
                                         \big( (\proposalset_{\mathcal{T}} \cup \proposalset_{\rho}^{+}) \setminus \proposalset_{\rho}^{-}\big)|\\
    = &  | \big((\proposalset_{M_1}\cap \proposalset_{\mathcal{T}})  \cup ( \proposalset_{M_1} \cap \proposalset_{\rho}^{+})  \big) \setminus (\proposalset_{M_1} \cap \proposalset_{\rho}^{-})|\\
    = &  | (\proposalset_{M_1}\cap \proposalset_{\mathcal{T}})| +w^{+}_{\rho} - w^{-}_{\rho}.
  \end{align*}
  The last equation holds because $\proposalset_{\mathcal{T}}\cap \proposalset^{+}_{\rho} = \emptyset$ and $\proposalset^{-}_{\rho} \subseteq \proposalset_{\mathcal{T}}$ and because of \cref{def:gain-lost}.
\end{proof}
\noindent By applying the above repeatedly, we obtain the following for an (arbitrary) closed
subset of rotations.
\begin{claim}\label{app:cor:SR-rotation-weights}
  If $C'$ is a closed subset of rotations and $\mathcal{T}$ is the preference table obtained by eliminating all rotations from $C'$ on the Phase~1 table,
  then $|\proposalset_{M_1}\cap \proposalset_{\mathcal{T}}| = |\proposalset_{M_1}\cap \proposalset_0| +
  \sum_{\rho\in C'}w^{+}(\rho) - \sum_{\rho\in C'}w^{-}(\rho)$.
\end{claim}

\noindent Now, we are ready to show our first statement.
To this end, let $C$ and $M_C$ be as defined in the lemma. 
By \cref{app:cor:SR-rotation-weights}, it follows that
\begin{align*} 
  |\proposalset_{M_1}\cap \proposalset_{M_C}|
  = &  |\proposalset_{M_1}\cap \proposalset_0| + \sum_{\rho\in C}w^{+}(\rho) - \sum_{\rho\in C}w^{-}(\rho)\\
   = & |\proposalset_{M_1}\cap \proposalset_0| + \big(\sum_{\sigma\in R_2}w^{+}(\sigma) - \sum_{\sigma\in R_2\setminus C}w^{+}(\sigma)\big) - \sum_{\rho\in C}w^{-}(\rho)\\
  \stackrel{\text{\scriptsize Lemma~\ref{lem:w(p)=-w(-p)}}}= & |\proposalset_{M_1}\cap \proposalset_0| + \big( \sum_{\sigma\in R_2}w^{+}(\sigma) - \sum_{\sigma\in R_2\setminus C}w^{-}(\neg\sigma)\big) - \sum_{\rho\in C}w^{-}(\rho)\\
  \stackrel{\text{\small $C$ is complete}}{=} & |\proposalset_{M_1}\cap \proposalset_0| + \sum_{\sigma\in R_2}w^{+}(\rho) - 2\cdot \sum_{\rho\in C}w^{-}(\rho).
\end{align*}
This completes the proof for the first statement.



\noindent The second statement follows directly from \eqref{eq:distanceByProposePairs}, \eqref{eq:intersectionByProposePairs}, and the first statement.

As for the last statement, we  first observe that for each two distinct rotations~$\sigma$ and $\sigma'$ from $R_2$ it holds that $\proposalset^{+}_{\sigma} \cap \proposalset^{+}_{\sigma'} = \emptyset$.
  To see this, let $\sigma\coloneqq ((e_0,h_0), \ldots, (e_{r-1}, h_{r-1}))$ which is exposed in some preference table~$\mathcal{T}$.
  Then, by definition, it holds that $\proposalset^{+}_{\sigma} = \{(e_i, h_{i+1}) \mid 0\le i \le r-1\}$.
  However, by \cite[Lemma 5.1, Corollary 5.1]{gusfield1988structure},
  for each~$i \in \{0,\ldots, r-1\}$, no other other rotation~$\rho$ exists such that
  $(e_i, h_{i+1}) \in \proposalset^{+}_{\rho}$.

  Next, we observe that $\proposalset_0\cap \proposalset^{+}_{\sigma}=\emptyset$.
  To see this, let $\sigma\coloneqq ((e_0,h_0), \ldots, (e_{r-1}, h_{r-1}))$.
  Then, $\proposalset^{+}_{\sigma} = \{(e_i,h_{i+1}) \mid 0\le i \le r-1\}$ ($i+1$ are taken modulo $r$).
  By the definition of rotations, it follows that $e_i$ prefers $h_i$ to $h_{i+1}$, implying that $h_{i+1}$ is \emph{not} the first agent in the preference list of $e_i$ right after Phase~1. That is, $(e_i, h_{i+1}) \notin \proposalset_{0}$.
  Thus, we have that
  \begin{align}
    |\proposalset_{M_1}\cap \proposalset_0| + \sum_{\sigma \in R_2}w^{+}(\sigma) = |\proposalset_{M_1}\cap \proposalset_0| + \sum_{\sigma \in R_2}|\proposalset_{M_1}\cap \proposalset_{\sigma}^+| =  & |\proposalset_{M_1}\cap \big( \proposalset_0 \cup \bigcup_{\sigma \in R_2} \proposalset^{+}_{\sigma}\big)|  \nonumber\\
    \le & 2\cdot |M_1|.   \label{eq:partial}
  \end{align}
  Now, to show the last statement, it suffices to show that $|M_1| \le |M_C|$.
  This is true by \cref{lem:rr-not-matched-in-M0-sound}:\\
  $|\proposalset_{M_1}\cap \proposalset_0| + \sum_{\sigma \in R_2}w^{+}(\sigma) - |M_1| - |M_C|
    \stackrel{\eqref{eq:partial}}{\le}  2\cdot |M_1| - |M_1| - |M_C|
    \stackrel{\text{\scriptsize Lemma~\ref{lem:rr-not-matched-in-M0-sound}}}{\le} 0$.
\end{proof}
}

This allows us to reduce the given instance (in polynomial time) to an instance of \WCFCS as follows; see \cref{ex:red-WCS} for an illustration.

\begin{construction}\label{cons:ISR-noties->WCFCS}
Finally, we arrive at the following instance of \WCFCS: 
\begin{enumerate}[(i)]
  \item Apply \cref{rr:not-matched-in-M0} in $O(n^2)$~time, and compute in $O(n^3\cdot \log{n})$ the rotation poset~$(R_2, \predr)$ for~$\mathcal{P}_2$~(see \cref{lem:correspondence+time}).
  \item  Compute~$\proposalset_0$ for the profile $\mathcal{P}_2$.
  Let $R\coloneqq R_2$ be the set of all dual rotations; note that $|R_2| \le {n \choose 2}$ (\cref{lem:bound-rotation-set}),
  \item Let $G$ be a graph on $R$ in which two elements of $R$ are adjacent if they form a dual pair of rotations (consequently, $G$ is a union of $\ell=|R|/2$ disjoint edges).
  \item The weight function~$w$ is defined by $w^-$.
  \item The budget~$b$ on the sum of weights is $b\coloneqq \frac{|\proposalset_{M_1}\cap \proposalset_0| + \sum_{\sigma \in R_2}w^{+}(\sigma) - |M_1| - |\proposalset_0|/2 + k}{2}$. 
\end{enumerate}
\end{construction}
Note that $b$ is derived from  \cref{lem:equivalence}(\ref{lem:weight-bound-equivalence}) such that we are searching for a complete and closed subset~$C$ of rotations whose sum of weights is bounded by $b$.
To see this, since$|\proposalset_{M}| = |\proposalset_0|/2$ (\cref{lem:rr-not-matched-in-M0-sound}), the budget $b$ is in fact equal to
$\frac{|\proposalset_{M_1}\cap \proposalset_0| + \sum_{\sigma \in R_2}w^{+}(\sigma) - |M_1| - |M| + k}{2}$, where $M$ is an arbitrary stable matching for $\mathcal{P}_2$.
Thus, by \cref{lem:equivalence}(\ref{lem:weight-bound-equivalence}), 
the budget~$b$ is at most $k/2$.
Note also that the sum of the weights $w^-(\rho)$ is upper-bounded by $n$ and thus the reduction presented above is a polynomial (many-to-one) reduction.

\begin{figure}[t]
      \centering
      \begin{tikzpicture}[>=stealth]
      [scale=1.0,auto]
      \node[nodest] (r1) at (0,0) {{$\rho_1$}};
      \node[nodest] (r2) at (2,0) {{$\rho_2$}};
      \node[nodest] (r3) at (0,-1) {$\rho_3$};
      \node[nodest] (r4) at (0,-2) {$\rho_4$};
      \node[nodest] (r5) at (2,-2) {{$\rho_5$}};
      \node[nodest] (r6) at (0,-3) {$\rho_6$};
      \draw[->,thick] (r1) -- (r3);
      \draw[->,thick] (r3) -- (r4);
      \draw[->,thick] (r4) -- (r6);
      \draw[->,thick] (r3) -- (r5);
      \draw[->,thick] (r2) -- (r5);

      \end{tikzpicture}
      \caption{A diagram for the partial order~$\predr$ on the rotation set~$R_2$ for the profile given \cref{example:tablesRotationsElimination}, as discussed in \cref{ex:red-WCS}.}
      \label{fig:diagram-for-poset}
\end{figure}
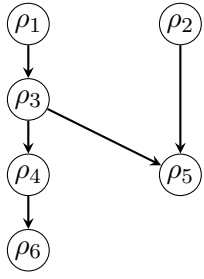

\begin{example}\label{ex:red-WCS}
  Consider the profile given in \cref{example:tablesRotationsElimination}.
Let~$M_1=\{\{1,7\},\{2,3\},\{4,6\},\{5,8\}\}$ and~$k=4$.
To find a stable matching~$M$ such that~$\dist(M_1,M) \le 4$, we first compute the rotation poset~$(R=\{\rho_1,\rho_2,\rho_3,\rho_4,\rho_5,\rho_6\}, \predr)$.
The partial order~$\predr$ is shown in \cref{fig:diagram-for-poset}.

Since~$\rho_1$ and~$\rho_3$ are singleton rotations, we can first eliminate them and get the set of all dual rotations~$R_2=\{\rho_2,\rho_4,\rho_5,\rho_6\}$, where~$\rho_2=\neg \rho_6$ and $\rho_4=\neg \rho_5$.
After the elimination of~$\rho_1$ and~$\rho_3$, we obtain the following table.

\begin{tabular}{l|lllllll}
    1 & 6 & 5 & 3 & 4 \\
    2 & 3 \\
    3 & 2 \\
    4 & 1 & 7 \\
    5 & 7 & 1 & 8 \\
    6 & 8 & 4 & 5 & 1 \\
    7 & 4 & 3 & 5 \\
    8 & 5 & 6 \\
\end{tabular}

Now the proposal set $\proposalset_0=\{(1,6),(2,3),(3,2),(4,1),(5,7),(6,8),(7,4),(8,5)\}$.
Since $M_1=\{\{1,7\},\{2,3\},\{4,6\},\{5,8\}\}$, we have that $\proposalset_{M_1}=\{(1,7),(7,1),(2,3),(3,2),(4,6),(6,4),(5,8),(8,5)\}$.
Then we can define weight functions~$w^+$ and~$w^-$.
For example, for rotation~$\rho_4=((1,6),(8,5))$:

\begin{align*}
\proposalset_{\rho_4}^+=&\{(1,5),(8,6)\}, \text{ and }  \\
\proposalset_{\rho_4}^-=&\{(1,6),(8,5)\}.
\end{align*}

Hence,
\begin{align*}
w^+(\rho_4)=&|\proposalset_{\rho_4}^+ \cap \proposalset_{M_1}|=|\{\emptyset\}|=0, \text{ and } \\
w^-(\rho_4)=&|\proposalset_{\rho_4}^- \cap \proposalset_{M_1}|=|\{(8,5)\}|=1.
\end{align*}
Similarly, we can compute weights for other rotations in~$R_2$:

\[
w^+(\rho_2)=0, w^-(\rho_2)=0, w^+(\rho_5)=1, w^-(\rho_5)=0, w^+(\rho_6)=0, w^-(\rho_6)=0.
\]

Now we can construct the instance~$(R=R_2,G=(R_2,E),w,\ell,b)$ of \WCFCS.
Edge set~$E$ consists of two edges~$\{\rho_2,\rho_6\}$ and $\{\rho_4,\rho_5\}$ since they are dual pairs of rotations.
The weight function~$w$ is defined as:
\[
w(\rho_2)=w^-(\rho_2)=0,w(\rho_4)=w^-(\rho_4)=1,w(\rho_5)=w^-(\rho_5)=0,w(\rho_6)=w^-(\rho_6)=0.
\]

Since~$|R_2|=4$, we have~$\ell=|R_2|/2=2$.
The budget~$b$ is given by:

\begin{align*}
b=& (|\proposalset_{M_1}\cap \proposalset_0| + \sum_{\sigma \in R_2}w^{+}(\sigma) - |M_1| - |\proposalset_0|/2 + k)/2 \\
=& (3+1-4-4+4)/2 \\
=& 0
\end{align*}

Now the task is to find a closed subset~$C \subseteq R$ of size~$\ell=2$ which is independent in~$G$ such that~$\sum_{c \in C} w(c)\le 0$.
It is easy to see that~$C=\{\rho_2,\rho_5\}$ is the only solution such that~$\sum_{c \in C} w(c)\le 0$.
By eliminating~$\rho_2$ and~$\rho_5$, we get matching~$M_{\{\rho_1,\rho_2,\rho_3,\rho_5\}}=\{(1,6),(2,3),(7,4),(8,5)\}$.
It is easy to check that~$\dist(M_1,M_{\{\rho_1,\rho_2,\rho_3,\rho_5\}})=4$.
\end{example}

\begin{figure}[t]
      \centering
      \begin{tikzpicture}[>=stealth]
      [scale=1.0,auto]

      \node[nodest] (r2) at (2,0) {$\rho_2$};
      \node[nodest] (r4) at (0,0) {$\rho_4$};
      \node[nodest] (r5) at (2,-2) {$\rho_5$};
      \node[nodest] (r6) at (0,-2) {$\rho_6$};
      \draw[thick] (r2) -- (r6);
      \draw[thick] (r4) -- (r5);

      \node (w4) [left=0.2cm of r4] {1};
      \node (w6) [left=0.2cm of r6] {0};
      \node (w2) [right=0.2cm of r2] {0};
      \node (w5) [right=0.2cm of r5] {0};
      \end{tikzpicture}
      \caption{The directed graph~$G$ on the rotation set~$R_2$ for as discussed in \cref{ex:red-WCS}, which shows the precedence relation between the rotations. The weights of the weight function~$w$ are depicted next to the vertices (i.e., rotations).}
      \label{fig:G-of-instance}
\end{figure}
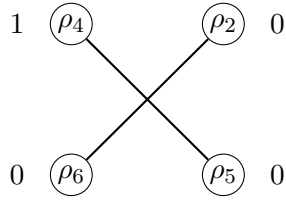

\paragraph{Solving \WCFCS when the conflict graph consists of $\ell$~disjoint cliques.}
Now, to prove \cref{thm:ISR FPT without ties}, we only need to show the following; recall that the budget~$b$, as defined in the construction, is at most $k/2$.
\begin{lemma}
\label{lem:algorithm}
If $G$ consists of exactly $\ell$ cliques, then \WCFCS can be solved in~$O((\Delta(G) + 1))^{b} \cdot |R|^2)$ time.
\end{lemma}
\appendixproof{lem:algorithm}{
\begin{proof}
  We present an algorithm for  the special case of \WCFCS when the conflict graph $G$ is a union of exactly $\ell$ cliques (see \cref{alg:1}).

First we introduce some notation.
For every element~$p \in R$, let
\[
T^\uparrow_p=\{q \in R \mid q \predr p\}
\qquad\qquad \textrm{and} \qquad\qquad
T^\downarrow_p= \{q \in R \mid p \predr q\}
\]
be the set of all predecessors and the set of all ancestors of~$p$ (including $p$ itself).
For a set $P \subseteq R$ we define $T^\uparrow_P$ as the union $\bigcup_{p \in P} T^\uparrow_p$; the set $T^\downarrow_P$ is defined analogously.

Before we continue with the proof, we observe the following.
\begin{observation}
\label{obs:T_p in S and T_-p' not in S}
Let~$S$ be a solution to \WCFCS.
If~$p \in S$, then~$T^\uparrow_p \subseteq S$ and~\mbox{$T^\downarrow_{N_G(p)} \cap S = \emptyset$}.
\end{observation}
\begin{proof}   \renewcommand{\qedsymbol}{(of \cref{obs:T_p in S and T_-p' not in S})~$\diamond$}
When~$p \in S$, since~$S$ is closed, we get~$T^\uparrow_p \subseteq S$.
Now, due to the conflicts $S \cap N_G(p) = \emptyset$ and consequently no element of $T^\downarrow_{N_G(p)}$ can be a part of $S$.
\end{proof}

According to \cref{obs:T_p in S and T_-p' not in S}, when we decide to include~$p$ into~$S$, we have to also include all elements in~$T^\uparrow_p$ into~$S$ and delete all elements in~$T^\uparrow_p$ and $T^\downarrow_{N_G(p)}$ from~$R$.
We extend the weight of a single element to the set~$T \subseteq R$ of elements in the natural way, that is, we set $w(T) \coloneqq \sum_{p \in T} w(p)$.

To upper-bound the weights of the sets $T^\uparrow_p$ from above, we present the following data reduction rules.

\begin{rrule}\label{rr:largeWeights}
    If $w(T^\uparrow_p) > b$, then we delete $T^\downarrow_p$ from $R$.
\end{rrule}
\begin{rrule}\label{rr:singleElementClique}
  If \cref{rr:largeWeights} is not applicable and there exists a clique with exactly one vertex~$p$, then we add $T^\uparrow_p$ to $S$ and decrease $b$ by its weight $w(T^\uparrow_p)$.
\end{rrule}
\begin{rrule}\label{rr:nonemptyIntersection}
  Let $K$ be a clique in $G$.
  If $p \in T^\uparrow_K$, then we add $T^\uparrow_p$ to $S$ and decrease~$b$ by its weight $w(T^\uparrow_p)$.
\end{rrule}

In fact it follows from the results of Gusfield~\cite{gusfield1988structure} that \cref{rr:nonemptyIntersection} is never used in any \WCFCS instance resulting from the construction we just defined.

\begin{claim}\label{lem:rrSoundness}
   \cref{rr:largeWeights,rr:singleElementClique,rr:nonemptyIntersection} are sound.
   Furthermore, \cref{rr:largeWeights,rr:singleElementClique,rr:nonemptyIntersection} can be implemented to run in $O(|R|^2)$ time.
\end{claim}
\begin{proof}    \renewcommand{\qedsymbol}{(of \cref{lem:rrSoundness})~$\diamond$}
  Note that we have to take at least one vertex from each clique in $G$ into~$S$ since we are about to take $\ell$ elements in total; thus, \cref{rr:nonemptyIntersection,rr:singleElementClique} are sound.
  Clearly, no solution $S$ can contain any set $T$ of weight $w(T) > b$, yielding soundness of \cref{rr:largeWeights}.

  Indeed, \cref{rr:largeWeights} runs in $O(R)$~time.
  As for \cref{rr:singleElementClique,rr:nonemptyIntersection}, we need to update~$T^\uparrow_q$, $T^\downarrow_q$ and the corresponding weights for every element $q$ left in the (updated) set of rotations $R$.
  This can be done in a straightforward way in $O(|R|^2)$ time.
\end{proof}


  \begin{algorithm}[t!]
\SetKwFunction{FindS}{FindS}
\SetKwProg{Fn}{}{}{}

\Fn{\FindS{$R, G, b, \ell$}}{
  \lIf{$b < 0$}{\Return No}
  $S \leftarrow \emptyset$\;
  \lWhile(\Comment{\cref{rr:largeWeights}}){$\exists p \in R$ with~$w(T^\uparrow_p)>b$}{
    delete $T^\downarrow_p$ from $R$ and $G$
    }
  \While(\Comment{\cref{rr:nonemptyIntersection}}){$\exists p \in R$ with~$p \in T^\uparrow_K$ for some clique $K$ in $G$}{
    $S \leftarrow S \cup T^\uparrow_{p}$ \;
    delete $T^\uparrow_p$ from $R$ and $G$ \;
    $b \leftarrow b - w(T^\uparrow_p)$ \;
  }
  \lIf{$G$ contains less than $\ell$ cliques}{\Return No}
  \While(\Comment{\cref{rr:singleElementClique}}){$G$ contains an isolated vertex $p$}{
    $S \leftarrow S \cup T^\uparrow_{p}$ \;
    delete $T^\uparrow_p$ from $R$ and $G$ \;
    $b \leftarrow b - w(T^\uparrow_p)$ \;
    $\ell \leftarrow \ell - 1$
  }

  \While{$\exists K$ clique in $G$ with~$w(T^\uparrow_p)>0$ for all $p \in K$}{
    \ForEach{$p$ in $K$}{
      $S' \leftarrow$ \FindS{$R \setminus \left( T^\uparrow_p \cup T^\downarrow_{K-p} \right), G, b - w(T^\uparrow_p), \ell - 1$} \;
      \lIf{$S' \neq$ No \textnormal{\textbf{and}} $S \cup T^\uparrow_p \cup S'$ is independent in $G$}{
      \Return $S \cup T^\uparrow_p \cup S'$}
    }
    \Return No \;
  }

  \lIf(\Comment{$R \neq \emptyset$}){$\exists p \in R$ in clique $K$ with $w(T^\uparrow_p) = 0$}{\Return $S \cup $ \FindS{$R \setminus \left( T^\uparrow_{p} \cup T^\downarrow_{K - p} \right), G, b, \ell - 1$}}

  \Return $S$ \;
}
\caption{}\label{alg:1}
\end{algorithm}

  We apply \cref{rr:largeWeights,rr:singleElementClique,rr:nonemptyIntersection} exhaustively.
  Now, we branch on a clique $K$ with~$0 < w(T^\uparrow_p) \le b$ for every $p \in K$ since we have to add one element of $K$ into any solution anyway.
  Note that in this branching we omit those cliques $K$ for which $\min_{p \in K} w(T^\uparrow_p) = 0$ holds.
  This procedure yields a search tree of depth at most $b$ with branching factor $\Delta(G) + 1$.
  Suppose that there is a leaf with nonnegative budget in which there is no clique left for branching.
  We claim that this is a yes-instance.
  At this point any clique left in $G$ contains at least one element $p$ with $w(T^\uparrow_p) = 0$; we select greedily these elements into the solution $S$.
  Observe that this gives a valid solution since all of these sets are of total weight $0$ (note that by including $p$ in the solution all of the weights can only decrease).

  As for the claimed running time, recall that it is possible to update all of the sets $T$ left for branching in $O(|R|^2)$ time.
  Furthermore, \cref{rr:singleElementClique,rr:nonemptyIntersection} are invoked $O(b)$ times and \cref{rr:largeWeights} at most once.
\end{proof}
}

We finally come to our main theorem.
\begin{proof}[Proof of  \cref{thm:ISR FPT without ties} ]
  By \cref{cons:ISR-noties->WCFCS}, in polynomial time,  we construct an instance for \WCFCS{} where the budget~$b$ is upper-bounded by $k/2$ and the conflict graph consists of $|R_2|/2$~edges~(see the reasoning right after \cref{cons:ISR-noties->WCFCS}).
  By \cref{lem:algorithm}, we can solve this instance and thus our problem in~$O(2^k \cdot n^4)$ time.
\end{proof}

\section{Hard cases of \ISR}
\label{sec:ISR instractable}

%

Throughout this section, we are using the following non-standard ``incremental'' (resp.\ ``decremental'')
variants of the \textsc{Independent Set} (resp.\ \textsc{Clique}) problem to show parameterized
intractability.
%

Our first problem asks for an independent set of size~$h$ for some graph in the case
when an independent set of size~$h$ for the graph minus an edge is already known.

\decprob{\EIIS}
{A graph~$G$, a distinguished edge~$e^* \in E(G)$, a positive integer~$h$, and an independent set~$S^*$ of size~$h$ for $G-e^*$.}
{Does there exist an independent set~$S$ of size~$h$ in $G$?
}

Our second problem asks for a clique with pendant edges of size~$h$ for some graph in the case
when a clique of size~$h$ with pendant edges for the graph with an additional edge is known.
A \emph{clique with pendant edges for a graph~$G$} is a subset~$V' \subset V(G)$ of vertices such that $V'$ forms a \emph{clique} in $G$, i.e., each two vertices in $V'$ are adjacent,
and that each vertex in $V'$ has at least one neighbor outside~$V'$.
The size of a \emph{clique with pendant edges} is defined as the number of vertices in the clique.

\decprob{\EDClPE}
{A graph~$G = (V,E)$, a distinguished edge~$e^* \in E(G)$, a positive integer~$h$,
and a clique~$S^*\subseteq V$ with pendant edges of size~$h$ for~$G$.}
{Is there a clique $S \subseteq V$ with pendant edges of size~$h$ in~$G - e^*$?
}

\begin{lemma}
  \label{lem:W1h-EIIS-EDClPE}
  \EIIS and \EDClPE are \np-hard
  and \woneh{} with respect to $h$.
\end{lemma}
\appendixproof{lem:W1h-EIIS-EDClPE}
{
\begin{proof}
%

  Note that \woneh{ness} under Turing-reductions is immediate for these problems and
  the simple tricks used to show the following lemma are to obtain hardness also under many-one reductions.

  \smallskip
  \noindent \textbf{\EIIS.}
  We give a reduction from \IS parameterized by the size~$h$ of the solution to \EIIS again with parameter $h$.
  Let $(\hat{G}, h)$ be an instance of \IS.
  We construct graph $G$ from $\hat{G}=(V,E)$ as follows: We add a set $S^*$ of $h$ new vertices and for each vertex $v\in V$ and each vertex $s \in S^*$ we add edge~$\{v, s\}$.
  We finish the construction by picking two distinct vertices~$s$ and $s'$ from $S^*$ and adding to $E$ the edge $e^* = \{s,s'\}$. 
  Clearly, $S^*$ is an independent set of size $h$ in $G - e^*$ as required.

  Observe that for each independent set $S$ in $G$ we either have $S \subseteq V$ or $S \subseteq S^*$, since we have added a complete bipartite graph between $V$ and $S$.
  Now, since any independent set $S$ in $G$ with $S \subseteq S^*$ contains at most $h-1$ vertices, an independent set of size $h$ in $G$ can only contain vertices from the set $V$, thus it must be an independent set of size $h$ in the graph $\hat{G}$.

  \smallskip
  \noindent  \textbf{\EDClPE.}
  Without loss of generality, we assume that~$h>2$ and describe a straight-forward parameterized reduction from \EIIS parameterized by the size~$h$ of the solution to \EDClPE with parameter~$h$.
  Given the instance~$(G,e^*,S^*,h)$ of \EIIS, we create the graph~$G'$ by complementing the graph~$G - e^*$   and adding for each original vertex in~$v\in V(G')$ a new vertex that is only connected to~$v$.
  It is easy to verify that~$(G',e^*,S^*,h)$ is a yes-instance of \EDClPE if and only if $(G, e^*, S^*, h)$ is a yes-instance of \EIIS.
  (The newly added vertices ensure the existence of the pendant edges and can never be part of a clique of size more than two.)
\end{proof}
}

\subsection{\ISR is generally hard even without ties}
To show that \ISR without ties is \np-hard, we identify a relation of it to an egalitarian variant of stable matching
where the egalitarian cost is minimized. 
Here, the \emph{egalitarian cost} of a matching~$M$ is defined as the sum of the ranks of the agents with respect to their partners,
and
the \emph{rank} of an agent~$x$ with respect to its partner~$M(y)$ is equal to the number agents that $x$ prefer over $y$.
Feder~\cite{feder1992new} showed that finding a stable matching with minimum egalitarian cost is \np-hard for \SR{}, even for complete preferences without ties (see the work of \citet{CHSYicalp-par-stable2018} for fixed-parameter tractability results on this problem).

The original hardness proof by \citet{feder1992new} is to reduce from \textsc{Vertex Cover}, which,
given an undirected graph~$G$ and an integer~$h' \in \mathds{N}$,
asks whether $G$ admits a \emph{vertex cover} of size~$h'$, i.e., a size-$h'$ vertex subset of~$V'\subseteq V(G)$ such that each edge in $E(G)$ is incident to at least one vertex from $V'$.
The basic idea behind the reduction is that putting a vertex to the solution set is equivalent to increasing the egalitarian cost by one.
This correspondence can also be achieved in \ISR by choosing an initial matching~$M_1$ which is associated to an empty vertex set,
thus finding a stable matching~$M_2$ closest to~$M_1$ is equivalent to finding a vertex cover of minimum size.
Note that \textsc{Vertex Cover} and \IS are dual to each other, i.e.,
a vertex subset~$V'$ is a vertex cover of size $h'$ if and only if $V(G)\setminus V'$ is an independent set of size $|V(G)|-h'$.
Due to this and since \IS is \woneh{} with respect to $|V|-h'$, we will directly reduce from \IS,
also showing parameterized intractability for our problem.


\begin{proposition}
\label{prop:InSR no tie NP-hard}
\ISR without ties is \np-hard and \woneh{} with respect to~$k'=|M_1\cap M_2|$.
\end{proposition}
\todo{Full proof.}

\appendixproof{prop:InSR no tie NP-hard}
{
\begin{proof}[Proof Sketch]
  We present a reduction from the \textsc{Independent Set} problem which is \woneh{} when parameterized by the solution size.
  Let $G = (V, E)$ be a graph and $h$~be the desired size of the independent set.

  Following essentially the construction of Feder~\cite{feder1992new}, we create our instance of \ISR as follows.
  For each vertex~$v_i \in V$, we create four agents $p_i, \bar{p}_i, q_i, \bar{q}_i$.
  The preferences of these agents in~$\mathcal{P}_2$ are identical to the profile constructed by Feder~\cite{feder1992new}:
  \begin{alignat*}{4}
    \agent~p_i\colon &  \bar{p}_i \succ [\{p_j \mid \{p_i,p_j\} \in E\}] \succ \bar{q}_i, && \agent~\bar{p}_i\colon&    q_i \succ p_i,   \\
  \agent~q_i\colon   & \bar{q}_i \succ \bar{p}_i,  &\qquad& \agent~\bar{q}_i\colon & p_i       \succ q_i.
\end{alignat*}
The old profile $\mathcal{P}_1$ is defined as follows.
  \begin{alignat*}{4}
    \agent~p_i\colon &  \bar{p}_i \succ \bar{q}_i \succ [\{p_j \mid \{p_i,p_j\} \in E\}], && \agent~\bar{p}_i\colon&    q_i \succ p_i,   \\
  \agent~q_i\colon   & \bar{q}_i \succ \bar{p}_i,  &\qquad& \agent~\bar{q}_i\colon & p_i       \succ q_i.
\end{alignat*}
  Following the proof of Feder~\cite{feder1992new} we can conclude that matching all $p_i$ with the $\bar{q}_i$ and $q_i$ with the $\bar{p}_i$ is a stable matching for $\mathcal{P}_1$.
  In order to obtain a stable matching for $\mathcal{P}_2$, however, one can only keep those $p_i$ matched with the $\bar{q}_i$ that correspond to an independent set.
\end{proof}
}


\subsection{Ties in the preferences make \ISM hard}
We show that \ISM becomes intractable when ties are allowed
even if the two preference profiles~$\mathcal{P}_1$ and~$\mathcal{P}_2$ are
almost identical.
The following result is achieved by a parameterized reduction from \EDClPE.

\begin{theorem}\label{thm:ISMTiesIncompleteHardness}
  \ISM with ties is \woneh{} with respect to $k$, even if $|\mathcal{P}_1 \oplus \mathcal{P}_2| = 1$.
\end{theorem}
\todohua{Hua: As already discussed in some earlier email, I think one can show that the problem remains W[1]-hard even if the input preferences are complete by appending to each preference list all remaining unlisted agents and by adding $\binom{h}{2}$ dummy agents to garbage collect those edge agents that are not matched. Since, this requires some more work, I'll leave this to later revision.}
\appendixproof{thm:ISMTiesIncompleteHardness}
{
\begin{proof}
  We present a parameterized polynomial-time reduction from the \woneh{} \EDClPE (see \cref{lem:W1h-EIIS-EDClPE}).
  Let $(G = (V, E), e^*, S^*, h)$ be an instance of \EDClPE, where $S^* \subseteq V$ is a clique with pendant edges of size~$h$ in~$G$.
  We assume without loss of generality that~$|E|>\binom{h}{2}+h$.

  To identify each clique vertex with its pendant edge, we introduce the following notion.
  Let $C$ be a clique with pendant edges.
  By definition, each clique vertex from $C$ has at least one neighbor outside $C$.
  Thus, for each clique vertex~$v\in C$, let $w$ be an arbitrary but fixed neighbor of $v$ outside $C$.
  We use $\pen_C(v)$ to denote the corresponding pendant edge~$\{v,w\}$.
  Further, we denote by $\pen(C)\coloneqq \{\pen(v) \mid v \in C\}$ the set of pendent edges of clique~$C$.

  In what follows, agent sets~$U$ and~$W$ will not have the same cardinality so that some agents need to remain unmatched.
  Note that for incomplete preferences with ties, using standard padding tricks,
  we can make~$U$ and~$W$ have the same carnality and simulate an agent
  being unmatched by being matched with dummy agents.
  Since using dummy agents to ``garbage collect'' the unmatched agents will increase the size of symmetric difference between our initial matching and the sought matching,
  we need to make sure that the number of unmatched agents is upper-bounded by a function in $k$.
  
  \paragraph*{Agents, preferences, and the symmetric difference bound~$k$.}
  In our \ISM instance, we will have one agent for each vertex from~$V$ and one agent for each edge from~$E$.
  For simplicity, we use the same symbols for vertices (edges) and the corresponding agents.
  Further, we introduce three sets of auxiliary agents:
  $X\coloneqq\{x_1,\dots,x_{|V| - h}\}$, $Y\coloneqq\{y_1,\dots,y_{|E| - \binom{h}{2} - h - 1}, y^*\}$, and~$Q\coloneqq\{q_1,\dots,q_t\}$,
  where $t$~is an even number which is to be determined later,
  and two special agents~$y^{\dagger}$ and $e^{\dagger}$.

  We partition our agents such that $U=V \cup Y \cup \{ e^{\dagger} \} \cup \{ q_1, q_3, \ldots, q_{t-1}\}$ forms one part and
  $W=X \cup E \cup \{y^{\dagger} \} \cup \{ q_2, q_4, \ldots, q_{t}\}$ forms the other part.
  The preferences of the agents in profile~$\mathcal P_1$ are set as follows.
  \begin{alignat*}{4}
  ~~&  \agent~v \in V\colon && (\left\{ e \in E \mid v \in e \right\}) \succ (X), &~~ &  \agent~x\in X\colon&& (V) \succ q_1,\\
  & \agent~y \in Y\setminus \{y^*\} \colon && (E), &&\agent~e \in E\!\setminus\! \{e^*\}\colon && Y \succ (\{u \in V\mid u \in e\}),\\
  \intertext{The path agents:}
  &  \agent~q_1\colon && (X) \succ q_2, && \agent~q_{t}\colon & & q_{t-1}\\  
  \intertext{For each~$i \in \{2,\ldots, t-1\}\colon$}
  &  \agent~q_{i}\colon && q_{i-1} \succ q_{i+1},\\
  \intertext{Finally, we list the preferences of the four special agents associated with the extra edge \mbox{$e^* = \{ u^*, v^*\}$}:}
  &\agent~y^*\colon &&  y^{\dagger} \succ E, && \agent~e^* \colon &&  e^{\dagger} \succ (Y \cup \{y^*\}) \succ (\{u^*,v^*\}), \\
  &\agent~e^{\dagger}\colon&&  y^{\dagger} \succ e^*, && \agent~y^{\dagger}\colon && e^{\dagger} \succ y^*.
  \end{alignat*}
  The only agent changing its mind in profile $\mathcal P_2$, compared to $\mathcal P_1$, is the agent $y^{\dagger}$ who switches the relative order of the (only) two potential partners in its preference list.
The preference list of agent~$y^{\dagger}$ in $\mathcal{P}_2$ becomes $y^* \succ e^{\dagger}$.
  We set the difference~$k$ between the two matchings~$M_1$ and~$M_2$ to $k\coloneqq h^2 + 5h + 4$ and set $t \coloneqq 2\cdot \lceil k/2 \rceil + 2$, which is strictly larger than $k$. 
  It remains to define the matching~$M_1$ that is stable for profile~$\mathcal P_1$ to complete the construction.
  Before we do this, we sketch the idea of the construction.

  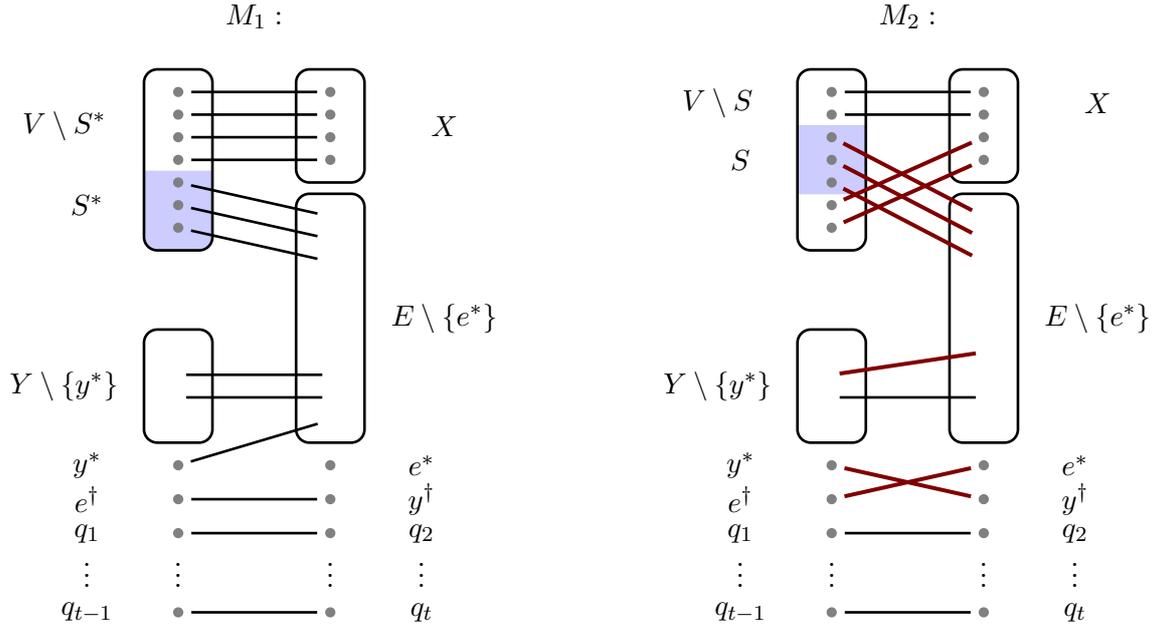
\begin{figure}[t]
  \centering
  \tikzstyle{bbox}=[draw, rounded corners=5pt]
  \tikzstyle{vertices}=[draw, circle, gray, fill=gray,inner sep=1pt]
  \begin{tikzpicture}[line width=1pt, scale=1]
    \def \xsc {0.3}

    \node at (1, 1) {$M_1:$};
    \foreach \i  in {1, 2, 3, 4} {
      \node[vertices] at (0, -\i*\xsc+\xsc) (v\i) {};
      \node[vertices] at (2, -\i*\xsc+\xsc) (x\i) {};
    }

    \foreach \i  in {5,6,7} {
      \node[vertices] at (0, -\i*\xsc+\xsc) (v\i) {};
    }
    \draw[bbox] (-1.5*\xsc,\xsc) rectangle (1.5*\xsc, -7*\xsc);

    \draw[bbox] (2-1.5*\xsc,\xsc) rectangle (2+1.5*\xsc, -4*\xsc);
    
    \gettikzxy{(v4)}{\fx}{\fy};
    \gettikzxy{(v5)}{\ffx}{\ffy};
    \begin{scope}
      \begin{pgfonlayer}{background}
      \clip[bbox] (-1.5*\xsc,\xsc) rectangle (1.5*\xsc, -7*\xsc);
      \draw[blue!20,fill=blue!20, even odd rule] (-1.5*\xsc,\xsc) rectangle (1.5*\xsc, -7*\xsc)
      (-1.5*\xsc, \xsc) rectangle (1.5*\xsc, \fy*0.5+\ffy*0.5);
      \end{pgfonlayer}
    \end{scope}
      
    \gettikzxy{(v2)}{\fx}{\fy};
    \gettikzxy{(v3)}{\ffx}{\ffy};
    \node (VS) at (-5*\xsc, \fy*0.5+\ffy*0.5) {$V\setminus S^*$};

    \node (X) at (2+5*\xsc, \fy*0.5+\ffy*0.5) {$X$};

    \gettikzxy{(v6)}{\fx}{\fy};
    \node (S) at (-4*\xsc, \fy) {$S^*$};
    
    \def \start {-10.5}
    \draw[bbox] (-1.5*\xsc,\start*\xsc) rectangle (1.5*\xsc, \start*\xsc-5*\xsc);
    \node (Ys) at (-5*\xsc, \start*\xsc-2.5*\xsc) {$Y\setminus \{y^*\}$};

    \foreach \i in {5,6,7,14} {
      \node[vertices] at (2, -\i*\xsc-0.5*\xsc) (e\i) {};
    }

    \draw[bbox,fill=white] (2-1.5*\xsc,-4.5*\xsc) rectangle (2+1.5*\xsc, \start*\xsc-5*\xsc);
    \node (Es) at (2+5*\xsc, 0.5*\start*\xsc-4.75*\xsc) {$E\setminus \{e^*\}$};
    \path[shorten <= 3pt, shorten >= 3pt] (0, \start*\xsc-2*\xsc) edge (2,  \start*\xsc-2*\xsc);
    \path[shorten <= 3pt, shorten >= 3pt] (0, \start*\xsc-3*\xsc) edge (2,  \start*\xsc-3*\xsc);

    \def \start {-16.5}
    \node[vertices] at (0, \start*\xsc) (ys) {};
    \node[] at (-4*\xsc, \start*\xsc)  {$y^*$};
    \node[vertices] at (2, \start*\xsc) (es) {};
    \node[] at (2+4*\xsc, \start*\xsc)  {$e^*$};

    \def \start {-18}
    \node[vertices] at (0, \start*\xsc) (ed)  {};
    \node[] at (-4*\xsc, \start*\xsc)  {$e^{\dagger}$};
    \node[vertices] at (2, \start*\xsc) (yd)  {};
    \node[] at (2+4*\xsc, \start*\xsc)  {$y^{\dagger}$};

    \def \start {-19.5}
    \node[vertices] at (0, \start*\xsc)  (q1) {};
    \node[] at (-4*\xsc, \start*\xsc)  {$q_1$};
    \node[vertices] at (2, \start*\xsc)  (q2) {};
    \node[] at (2+4*\xsc, \start*\xsc)  {$q_2$};

    \def \start {-21}
    \node[] at (0, \start*\xsc)  (vds) {$\vdots$};
    \node[] at (-4*\xsc, \start*\xsc)  {$\vdots$};
    \node[] at (2, \start*\xsc)  (vdss) {$\vdots$};
    \node[] at (2+4*\xsc, \start*\xsc)  {$\vdots$};

    \def \start {-23}
    \node[vertices] at (0, \start*\xsc)  (qt1) {};
    \node[] at (-4*\xsc, \start*\xsc)  {$q_{t-1}$};
    \node[vertices] at (2, \start*\xsc)  (qt) {};
    \node[] at (2+4*\xsc, \start*\xsc)  {$q_{t}$};

    \foreach \s / \t in {q1/q2, qt1/qt, ed/yd, v1/x1,v2/x2,v3/x3,v4/x4,v5/e5,v6/e6,v7/e7, ys/e14} {
      \path[shorten <= 3pt, shorten >= 3pt] (\s) edge (\t);
    }

\end{tikzpicture}\quad\qquad\qquad
  \begin{tikzpicture}[line width=1pt, scale=1]
    \def \xsc {0.3}

    \node at (1, 1) {$M_2:$};
    \foreach \i  in {1, 2, 3, 4} {
      \node[vertices] at (0, -\i*\xsc+\xsc) (v\i) {};
      \node[vertices] at (2, -\i*\xsc+\xsc) (x\i) {};
    }

    \foreach \i  in {5,6,7} {
      \node[vertices] at (0, -\i*\xsc+\xsc) (v\i) {};
    }
    \draw[bbox] (-1.5*\xsc,\xsc) rectangle (1.5*\xsc, -7*\xsc);

    \draw[bbox] (2-1.5*\xsc,\xsc) rectangle (2+1.5*\xsc, -4*\xsc);
    
    \gettikzxy{(v2)}{\fx}{\fy};
    \gettikzxy{(v3)}{\ffx}{\ffy};

    \gettikzxy{(v5)}{\gx}{\gy};
    \gettikzxy{(v6)}{\ggx}{\ggy};
    \begin{scope}
      \begin{pgfonlayer}{background}
      \draw[blue!20,fill=blue!20] (-1.5*\xsc,\fy*0.5+\ffy*0.5) rectangle (1.5*\xsc, \gy*0.5+\ggy*0.5);
      \end{pgfonlayer}
    \end{scope}
    
    \gettikzxy{(v1)}{\fx}{\fy};
    \gettikzxy{(v2)}{\ffx}{\ffy};
    \node (VS) at (-5*\xsc, \fy*0.5+\ffy*0.5) {$V\setminus S$};

    \node (X) at (2+5*\xsc, \fy*0.5+\ffy*0.5) {$X$};

    \gettikzxy{(v4)}{\fx}{\fy};
    \node (S) at (-4*\xsc, \fy) {$S$};
    
    \def \start {-10.5}
    \draw[bbox] (-1.5*\xsc,\start*\xsc) rectangle (1.5*\xsc, \start*\xsc-5*\xsc);
    \node (Ys) at (-5*\xsc, \start*\xsc-2.5*\xsc) {$Y\setminus \{y^*\}$};

    \foreach \i in {5,6,7} {
      \node[vertices] at (2, -\i*\xsc-0.5*\xsc) (e\i) {};
    }

    \draw[bbox,fill=white] (2-1.5*\xsc,-4.5*\xsc) rectangle (2+1.5*\xsc, \start*\xsc-5*\xsc);
    \node (Es) at (2+5*\xsc, 0.5*\start*\xsc-4.75*\xsc) {$E\setminus \{e^*\}$};
    \path[shorten <= 3pt, shorten >= 3pt, red!50!black, line width=1.5pt] (0, \start*\xsc-2*\xsc) edge (2,  \start*\xsc-\xsc);
    \path[shorten <= 3pt, shorten >= 3pt] (0, \start*\xsc-3*\xsc) edge (2,  \start*\xsc-3*\xsc);

    \def \start {-16.5}
    \node[vertices] at (0, \start*\xsc) (ys) {};
    \node[] at (-4*\xsc, \start*\xsc)  {$y^*$};
    \node[vertices] at (2, \start*\xsc) (es) {};
    \node[] at (2+4*\xsc, \start*\xsc)  {$e^*$};

    \def \start {-18}
    \node[vertices] at (0, \start*\xsc) (ed)  {};
    \node[] at (-4*\xsc, \start*\xsc)  {$e^{\dagger}$};
    \node[vertices] at (2, \start*\xsc) (yd)  {};
    \node[] at (2+4*\xsc, \start*\xsc)  {$y^{\dagger}$};

    \def \start {-19.5}
    \node[vertices] at (0, \start*\xsc)  (q1) {};
    \node[] at (-4*\xsc, \start*\xsc)  {$q_1$};
    \node[vertices] at (2, \start*\xsc)  (q2) {};
    \node[] at (2+4*\xsc, \start*\xsc)  {$q_2$};

    \def \start {-21}
    \node[] at (0, \start*\xsc)  (vds) {$\vdots$};
    \node[] at (-4*\xsc, \start*\xsc)  {$\vdots$};
    \node[] at (2, \start*\xsc)  (vdss) {$\vdots$};
    \node[] at (2+4*\xsc, \start*\xsc)  {$\vdots$};

    \def \start {-23}
    \node[vertices] at (0, \start*\xsc)  (qt1) {};
    \node[] at (-4*\xsc, \start*\xsc)  {$q_{t-1}$};
    \node[vertices] at (2, \start*\xsc)  (qt) {};
    \node[] at (2+4*\xsc, \start*\xsc)  {$q_{t}$};

    \foreach \s / \t in {q1/q2, qt1/qt, v1/x1, v2/x2} {
      \path[shorten <= 3pt, shorten >= 3pt] (\s) edge (\t);
    }
    \foreach \s / \t in {ed/es, v6/x3, v7/x4, v3/e5, v4/e6,v5/e7, ys/yd} {
      \path[shorten <= 3pt, shorten >= 3pt, red!50!black, line width=1.5pt] (\s) edge (\t);
    }

  \end{tikzpicture}
  \caption{Illustration of the stable matching~$M_1$ for $\mathcal P_1$ and a desired stable matching~$M_2$ for $\mathcal P_2$ (see \cref{def:M_1} and \cref{def:M_2}). Observe that each agent from $U$ is matched in both $M_1$ and $M_2$. The agents from $S$ and from $S^*\setminus S$ changed their partners. The agents from $\pen(S)$
    and from $E\setminus (\pen(S)\cup \binom{S}{2}) \setminus (E\setminus (\pen(S^*)\cup \binom{S^*}{2}))$ changed their partners.}
\end{figure}

  \paragraph{Construction idea and the initial matching~$\mathbf{M_1}$.}
  Observe that $|W|-|U|=\binom{h}{2}$ agents will remain unmatched.
  As already discussed at the beginning of the proof, we can introduce $\binom{h}{2}$ dummy agents to ``garbage collect'' these unmatched agents.
  However, this will increase our parameter~$k$ only by $2\cdot \binom{h}{2}$.
  To convey the actual idea of the reduction, however, we omit introducing these dummy agents.

  Now, the idea of our construction is that the set~$\bot(M_1)$ of agents that are not matched
  by our initial matching~$M_1$ corresponds to the set of edges of the clique~$S^*$ for~$G$.
  Analogously, for every stable matching~$M_2$ for $\mathcal{P}_2$ the set~$\bot(M_2)$ of agents that are not matched by~$M_2$ will have to correspond to a set of edges of some clique~$S$ for~$G-e^*$.

  This can be formally captured as follows.
  \begin{definition}\label{def:M_1}  We define the initial matching~$M_1$ as follows.
  \begin{enumerate}[(i)]
    \item Match agent~$e^{\dagger}$ with agent~$y^{\dagger}$. 
    \item Match every clique vertex agent~$v \in S^*$ with its pendant edge agent~$\pen_{S^*}(v)$.
    \item Match every non-clique vertex agent~$v \in V\setminus S^*$ to an arbitrary but fixed agent from~$X$. 
    \item\label{eq:Y} Match every agent from~$Y$ with an arbitrary but fixed non-pendant and ``non-clique'' edge agent~$e \in E\setminus ( \binom{S^*}{2} \cup \pen(S^*))$.
    \item For every odd number~$i \in \{1, 3, \ldots, {t-1}\}$, match agent~$q_i$ with agent~$q_{i+1}$.
  \end{enumerate}
\end{definition}

  This completes the construction which can be computed in polynomial time.
  The following claim implies that~$M_1$ is stable for $\mathcal{P}_1$.
  It will be used later in the correctness proof.
  \begin{claim} \label{cl:ISMTiesIncompleteHardness_matchingstable}
    The set of unmatched agents in $M_1$ is $\binom{S^*}{2}$.
    Moreover, $M_1$ is stable for $\mathcal{P}_1$.
  \end{claim}
  \begin{proof}\renewcommand{\qedsymbol}{(of \cref{cl:ISMTiesIncompleteHardness_matchingstable})~$\diamond$}
    The first statement follows from the fact that by \cref{def:M_1}, every agent from $U$ is matched and by \cref{def:M_1}\eqref{eq:Y},
    exactly those edge agents which correspond to the edges from $\binom{S^*}{2}$ are unmatched under $M_1$. 
    As for the second statement, suppose for the sake of contradiction, that $\{u,w\}$ is a blocking pair of $M_1$ with $u\in U$ and $w\in W$.
    Now, observe that each agent from
    $S^* \cup (Y \setminus \{y^*\}) \cup \{e^{\dagger}\} \subseteq U$ already obtains its most preferred agent as a partner.
    Moreover, for each $i \in \{1, \ldots, t/2\}$, the path agent~$q_{2i-1}$ only prefers to an agent~$z$ to its assigned partner with $z\in X\cup \{q_{2i'}\mid 1\le i' \le t/2\}$, but agent~$z$ already obtains its most preferred agent as a partner (either some vertex agent or some path agent).
    It follows that $u\in V\setminus S^*$ or $u=y^*$.
    Since $y^*$ only prefers $y^{\dagger}$ to its partner but $y^{\dagger}$ already obtains its most preferred agent as a partner, it follows that $u\in V\setminus S^*$, i.e.,  $u$ is some non-clique vertex agent.
    This implies that $M_1(v) \in X$.
    By the preferences of the agents from $V$ and since $M_1(u)\in X$, it follows that $w$ is some edge agent such that $u\in w$.
    Since $u$ is not a clique vertex from $S^*$, it follows that $w\not \subseteq \binom{S^*}{2}$, implying that $w\in E\setminus \binom{S^*}{2}$.
    If $w\in E\setminus (\binom{S^*}{2}\cup \pen(S^*))$, i.e., it is also \emph{not} a pendant edge of some clique vertex from $S^*$, then it obtains a partner from $Y$ and will not form with $u$ a blocking pair because $u$ is a vertex agent.
    Thus, $w$ must be a pendant edge agent from $\pen(S^*)$.
    In this case, $w$ is matched to another incident vertex agent which is different from $u$.
    However, since $w$ considers both its incident vertex agents as tied, it cannot form with $u$ a blocking pair, a contradiction.    
  \end{proof}

  \paragraph{Correctness of the Construction.}
   We show that there is a stable matching $M_2$ for profile~$\mathcal P_2$ with $|M_1 \Delta M_2|\le k$ if and only if
   there is a clique~$C'$ with pendant edges of size~$h$ for~$G-e^*$.

   For the ``only if'' direction,
   let $S$~be a clique with pendant edges~$\pen_{S}$ of size~$h$ for~$G-e^*$.
   For ease of notation let $E(S)=\binom{S}{2}\cup \pen(S)$ and $E(S^*)=\binom{S^*}{2}\cup \pen(S^*)$.   \begin{definition}\label{def:M_2}
   We construct the target matching~$M_2$ as follows.
     \begin{enumerate}[(i)]
     \item Match~$e^{\dagger}$ with $e^{*}$ and $y^*$ with $y^{\dagger}$.
     \item\label{eq:M2-pendant} Match each clique vertex~$v \in S$ with its pendant edge agent~$\pen_{S}(v)$.
     \item Match each non-clique vertex agent~$v\in V \setminus (S\cup S^*)$ with the partner~$M_1(v)$ (note that this partner is from $X$).
     \item Match each remaining non-clique vertex agent~$v\in S^*\setminus S$ with an arbitrary but fixed not-yet-matched agent from $X$.
     \item\label{eq:M2-non-clique-non-pendant} For each $y \in Y\setminus \{y^*\}$ with $M_1(y)\in E\setminus E(S)$,
     match $y$ with $M_1(y)$; recall that $M_1(y)\in E\setminus E(S^*)$.
     \item\label{eq:M2-non-clique-non-pendant2} Match each remaining not-yet-matched agent~$y\in Y\setminus \{y^*\}$ with an arbitrary but fixed not-yet-matched agent from $E\setminus E(S)$.
    \item For every odd number~$i \in \{1, 3, \ldots, {t-1}\}$, match agent~$q_i$ with agent~$q_{i+1}$.
   \end{enumerate}
 \end{definition}
 
  We claim the following for matching~$M_2$.
  \begin{claim} \label{cl:ISMTiesIncompleteHardness_matching2stable}
    The set of unmatched agents in $M_2$ is $\binom{S}{2}$.
    Moreover, $M_2$ is stable for $\mathcal{P}_2$.
  \end{claim}
  \begin{proof}\renewcommand{\qedsymbol}{(of \cref{cl:ISMTiesIncompleteHardness_matching2stable})~$\diamond$}
    Since $S$ is a clique with pendant edges for $G-e^*$, we can infer that
    $e^* \notin \binom{S}{2} \cup \pen(S)$.
    Thus, $M_2$ is indeed a matching.
    Moreover, by the definition of $M_2$ (see \cref{def:M_2}\eqref{eq:M2-pendant}, \eqref{eq:M2-non-clique-non-pendant}, \eqref{eq:M2-non-clique-non-pendant2}) and similarly to the proof for \cref{cl:ISMTiesIncompleteHardness_matchingstable}, we can conclude that exactly those edge agents  which correspond to some edge from $\binom{S}{2}$ are unmatched under $M_2$.

    As for the stability proof, observe that $y^{\dagger}$'s most preferred agent in $\mathcal{P}_2$ is agent~$y^*$.
    Using a proof similar to the one for \cref{cl:ISMTiesIncompleteHardness_matchingstable}, we can conclude that $M_2$ is stable for $\mathcal{P}_2$.
  \end{proof}
   
  Now, we upper-bound the size~$|M_1 \Delta M_2|$.
  \begin{claim}\label{claim:sym-diff-bound}
     It holds that $|M_1 \Delta M_2| \le k$.
  \end{claim}
    \begin{proof}\renewcommand{\qedsymbol}{(of \cref{claim:sym-diff-bound})~$\diamond$}
      To prove this, note that we only need to count the number of agents from $U$ who change their partners because both $M_1$ and $M_2$ match all agents from $U$ and unmatch exactly $\binom{h}{2}$ agents from $E$ (see \cref{cl:ISMTiesIncompleteHardness_matchingstable} and \cref{cl:ISMTiesIncompleteHardness_matching2stable}).
   Thus,  $|M_1\Delta M_2| \le 2 \cdot |\{u\in U \mid M_1 (u)\neq M_2(u)\}|$.
   Between $M_1$ and $M_2$, the following agents from $U$ may have changed their partners.
   \begin{enumerate}[(1)]
     \item The clique vertex agents from~$S$, 
     \item\label{bound-non-clique} the non-clique vertex agents from $S^* \setminus S$,
     \item\label{bound-Y} the agents~$y$ from $Y\setminus \{y^*\}$ with $M_2(y)\in
     \big(E\setminus E(S)\big) \setminus \big( E\setminus E(S^*)\big)$,
     \item the agents~$y^{*}$, and $e^{\dagger}$.
   \end{enumerate}
   Before we continue with the bound, let us recall a property in set theory.
   \begin{align}
 \text{ For each three sets~} A, B, C \text{ it holds that }     A\setminus (B\setminus C) = (A\cap C) \cup (A\setminus B).
   \end{align}
   Substituting $A=E\setminus E(S)$, $B=E $, and $C=E(S^*)$, we conclude the following for the set of agents described in \eqref{bound-Y}.
   \begin{align*}
     (E \setminus E(S)) \setminus (E\setminus E(S^*)) = \big( (E \setminus E(S) \cap E \big) \cup \big( (E\setminus E(S)) \setminus E\big)
     = (E\setminus E(S)) \cap E(S^*) \subseteq E(S^*).
   \end{align*}
   Summarizing, by the definition of $E(S^*)$, the number of agents from $U$ that changed their partners is at most $|S|+|S^*|+|\binom{S^*}{2}|+ |\pen(S^*)|+2 = 3h+\binom{h}{2}+2$.
   Thus, $|M_1 \Delta M_2| \le 2\cdot (3h+\binom{h}{2}+2) = h^2+5h+4$.
 \end{proof}

   For the ``if'' direction,
   let $M_2$~be an arbitrary stable matching for~$\mathcal P_2$ with $|M_1 \Delta M_2| \le k$.
   Our correctness proof mainly relies on the following claim.
   \begin{claim} \label{cl:ISMTiesIncompleteHardness_matchingimpliesclique}
     Let $Z \subseteq V$ be the set of vertex agents that are not matched to~$X$ in~$M_2$
     and let~$F =  E \cap \bot(M_2)$ be the set of unmatched edge agents. 
     The following holds.
     \begin{enumerate}[(i)]
       \item $e^* \notin F$.\label{eq:notinF}
       \item $|Z|=h$.\label{eq:size-Z}
       \item $F = \binom{Z}{2}$.\label{eq:F=binom{Z}{2}}
       \item Every agent from~$Z$ is matched to an edge agent from $E \setminus \binom{Z}{2}$. \label{eq:pendant-edges}
       \end{enumerate}
     \end{claim}
\begin{proof}\renewcommand{\qedsymbol}{(of \cref{cl:ISMTiesIncompleteHardness_matchingimpliesclique})~$\diamond$}
  For the first statement,  observe that by the preferences of the agents from $\{y^*, e^{\dagger}, y^{\dagger}, e^*\}$ in $\mathcal{P}_2$ and by the stability of $M_2$ for $\mathcal{P}_2$ it must hold that $M_2(e^{\dagger}) = e^*$ and $M_2(y^{\dagger}) = y^*$.
  Thus, $e^{*}$ is matched, implying that $e^{*}\notin F$.
  
  For the second statement, note that $M_2$ must match every agent from~$X$ to some agent from~$V$ as otherwise there would be an agent $x \in X$ with $M_2(x) = q_1$.
  Every stable matching, however, that contains~$\{x,q_1\}$ must also match $q_{2i}$ with $q_{2i+1}$ for every number~$i \in \{1, 2, \ldots, {t/2-1}\}$.
  The symmetric difference between $M_1$ and $M_2$ is at least $2\cdot (t/2-1 + 1) = t > k$---a contradiction.
  Since $|X|=|V|-h$ it follows that exactly $h$~agents from $V$ remain that are \emph{not} matched to~$X$.
  By the definition of $Z$, we immediately have
    $|Z|=h$. 

    To show the third statement, first, we show that
    \begin{align} F \subseteq \binom{Z}{2}. \label{eq:F<=binom{Z}{2}}\end{align}
    Towards a contradiction, suppose that there is an edge~$e \in F \setminus \binom{Z}{2}$, i.e.,
    there exists an unmatched edge agent~$e \in F$ for which at least one of its endpoint is matched to~$X$.
    Let $v \in e$ be such an incident vertex of $e$ with $M_2(v)\in X$. 
    Note that by the first statement, we have that $e \neq e^*$ since $e$ is unmatched under $M_2$.
    Now, we claim that $M_2$ is not stable since $\{v,e\}$ is a blocking pair.
    Indeed, edge agent~$e$ is not matched and prefers to be matched with~$v$ and $v$ finds~$e$ better than its current partner (from~$X$).
    Thus, $M_2$~is not stable---a contradiction.

    Second, we show that
    $|F| \ge \binom{h}{2}$. 
    To see this, observe that by the preferences of the agents from $E\setminus \{e^*\}$,
    each edge agent from $E\setminus \{e^*\}$ is either unmatched or matched to a vertex agent from $V$
    or matched to an agent from $Y\setminus \{y^*\}$.
    Now, recall by the definition of $Z$ that each agent from $V\setminus Z$ is matched to some agent from $X$.
    Thus, each edge agent from $E\setminus \{e^*\}$ is either unmatched or matched to an agent from $Z\cup Y\cup \{y^*\}$.
    By the definition of~$F$ and since $M_2(e^*)=e^{\dagger}$, it follows that
    \begin{align*}
      |F| \ge |E\setminus \{e^*\}| - |Z\cup (Y\setminus \{y^*\})| = |E|-1 -  (h+|E|-\binom{h}{2}-h-1) = \binom{h}{2}.
    \end{align*}
    The second to last equation holds because $|Z|=h$ (see the second statement).
    Together with \eqref{eq:F<=binom{Z}{2}}, we infer that $|F|=\binom{h}{2}$ and thus $F=\binom{Z}{2}$.

    Finally, to show the last statement, consider an arbitrary agent~$z\in Z$.
    By the definition of $Z$, it follows that $z$ is either unmatched or matched to some agent from $E\setminus F$.
    By the definition of $F$ and by the third statement, there exists an unmatched edge agent~$e\in F$
    which is incident to $z$.
    By the stability of $M_2$ it must hold that $z$ is matched to some edge agent from $E\setminus F$, i.e.,  from $E\setminus \binom{Z}{2}$.
    %
    %
  \end{proof}
  We construct our clique~$S$ with pendant edges~$\pen_S$ of size~$h$ for~$G-e^*$ by applying \cref{cl:ISMTiesIncompleteHardness_matchingimpliesclique} as follows.
  The clique is set to $S\coloneqq Z$ and for each vertex~$v\in S$ we set its pendant edge~$\pen_S(v)$ to $M_2(v)$.
  Note that by \cref{cl:ISMTiesIncompleteHardness_matchingimpliesclique}\eqref{eq:notinF} it holds that $e^* \notin F$.
  Together with and \cref{cl:ISMTiesIncompleteHardness_matchingimpliesclique}\eqref{eq:F=binom{Z}{2}}
  we infer that $|\{u^*,v^*\}\cap Z| \le 1$ and, hence, $Z$~is indeed a clique in $G-e^*$.
  By \cref{cl:ISMTiesIncompleteHardness_matchingimpliesclique}\eqref{eq:pendant-edges}, the defined pendant edges are all from $G-e^*$.
\end{proof}
}

In the following, we show that with respect to the number of common pairs between the target stable matching and the initial stable matching
the problem is parameterized intractable, even if the two input profiles differ by only two swaps.
The corresponding parameterized reduction is from \textsc{Independent Set}.

\todo[inline]{Update the constructed instance ($k$ and $k'$) in the proof for \cref{thm:ISMTiesDualParameter}---Hua}

\begin{theorem}\label{thm:ISMTiesDualParameter}
  \ISM with ties is \woneh{} 
  with respect to~$k'=|M_1 \cap M_2|$ of common pairs,
  even if\/ $\left| \mathcal{P}_1 \oplus \mathcal{P}_2 \right| = 2$.
\end{theorem}

\appendixproof{thm:ISMTiesDualParameter}
{
\begin{proof}
  We show this by reducing from the \woneh{} \textsc{Independent Set} problem, parameterized by the solution size.
  The parameterized \textsc{Independent Set} problem has, as input, an undirected graph~$G$ and an integer~$h \in \mathds{N}$--the parameter--and asks whether $G$ admits an \emph{independent set} of size~$h$, i.e.\ an $h$-vertex subset of~$V'\subseteq V(G)$ of pairwise nonadjacent vertices.
  Let $I=(G,h)$ be an instance of \textsc{Independent Set}.
  Further, let $V(G)=\{v_1,\ldots, v_n\}$ and $E(G)=\{e_0,\ldots, e_{m-1}\}$ denote the set of vertices and the set of edges in $G$ respectively; note that we start the index of the edges with zero to simplify our reasoning later.
  We construct an instance~$(P_1, P_2, M_1)$ of \ISM{} with ties,
  where the two profiles~$P_1$ and $P_2$ are preference profiles for two disjoint sets of agents, $U\cup W \cup E \cup F \cup A \cup B$ and $X\cup Y \cup H \cup C \cup D$,
  and the matching~$M_1$ is stable for $P_1$ such that $P_1$ and $P_2$ will differ from each other by only two swaps.
  We will be searching for stable matching~$M_2$ for $P_2$ with $|M_1 \cap M_2| \ge 2h$.

  \paragraph{The agent sets.}
  For each vertex~$v_i \in V(G)$, we introduce four agents~$u_i,w_i,x_i,y_i$,
  and add them to $U$, $W$, $X$, $Y$, respectively.

  For each edge~$e_\ell \in E(G)$ with two endpoints~$v_i$ and $v_j$, do the following.
  Introduce four agents~$e_\ell$, $f_\ell$, $h_\ell^{u_i}$, and $h_\ell^{u_j}$.
  Add $e_\ell$ to $E$, $f_\ell$ to $F$, and $h_{\ell}^{u_i}$ as well as $h_\ell^{u_j}$ to $H$.
  Additionally, introduce the auxiliary agents~$a_\ell$, $b_\ell$, $c_\ell$, and $d_\ell$,
  and add them to $A$, $B$, $C$, and $D$, respectively.

  \paragraph{The preference lists of the agents from $U$, $W$, $X$, and $Y$.}
  The preference lists of these agents are the same in both $P_1$ and $P_2$,
  and are constructed in such a way that an arbitrary stable matching must match these agents among themselves.
  Here, $[\star]$ means that the elements in~$\star$ are ranked in an arbitrary but fixed order, while~$(\star)$ means that the elements in~$\star$ are tied.
  \begin{alignat*}{4}
    \forall i \in [n], &\quad & u_i \colon & x_i \succ [ \{ h_\ell^{u_i} \mid  e_\ell \in E(G) \text{ with } v_i \in e_{\ell}\}] \succ y_i, &\qquad & x_i \colon && w_i \succ u_i,\\
    && w_i \colon & y_i \succ x_i && y_i \colon && u_i \succ w_i.\
  \end{alignat*}

  \paragraph{The preference lists of the agents from $E$, $F$, and $H$.}
  The preference lists of these agents are the same in both $P_1$ and $P_2$.
   \begin{alignat*}{4}
     \forall e_\ell \in E[G]\text{ with } e_\ell=\{v_i,v_j\} \text{ and } i < j, &\quad &e_{\ell} &\colon (h^{u_i}_{\ell}, h^{u_j}_{\ell}) \succ  c_\ell, & \qquad & h_\ell^{u_i} \colon && (e_\ell, u_i, a_\ell) \succ f_\ell, \\
     & &f_{\ell} &\colon (h^{u_i}_{\ell}, h^{u_j}_{\ell}) \succ  d_\ell, & \qquad & h_\ell^{u_j} \colon && (e_\ell, u_j, b_\ell) \succ f_\ell.
   \end{alignat*}

   \paragraph{The preference lists of the auxiliary agents from $A$, $B$, $C$, and $D$.}
   Except for the agents $c_0$ and $d_0$, the preference lists of these agents are the same in both $P_1$ and $P_2$.
   We first describe the preference lists of all the auxiliary agents in $P_1$, and then only describe the preference lists of $c_0$ and $d_0$ in $P_2$ as the other will remain the same.
  The operations~``$\ell+1$'' and ``$\ell-1$'' are taken modulo $m$.
    \begin{alignat*}{4}
      \forall e_\ell \in E[G]\text{ with } e_\ell=\{v_i,v_j\} \text{ and } i < j, &\quad &
      a_{\ell} &\colon  c_{\ell+1} \succ h_\ell^{u_i}, & \qquad & c_\ell \colon && (a_{\ell-1}, e_\ell), \\
     && b_{\ell} &\colon  d_{\ell+1} \succ h_\ell^{u_j}, & \qquad & d_\ell \colon && (b_{\ell-1}, f_\ell).
   \end{alignat*}
  In $P_2$, the preference lists of agents~$c_0$ and $d_0$ will be changed so that the acceptable agents are not tied any more:
  \begin{alignat*}{3}
     c_0\colon & a_{m-1} \succ e_{0}, \text{ and } d_0 \colon & b_{m-1} \succ f_0.
   \end{alignat*}

   It is straight-forward to verify that $P_1$ and $P_2$ differ by only two swaps.

   \paragraph{The initial stable matching~$M_1$ of $P_1$.}
   For each vertex~$v_i \in V(G)$, let $M_1(u_i)=y_i$, $M_1(w_i)=x_i$.
   For each edge~$e_\ell \in E(G)$ with $e_\ell=\{v_i, v_j\}$ and $i < j$,
   let $M_1(e_\ell)=c_\ell$, $M_1(f_\ell)=d_\ell$, $M_1(a_\ell)=h_\ell^{u_i}$, and $M_1(b_\ell)=h_\ell^{u_j} $.
   One can verify that this matching is indeed a stable matching of~$P_1$.

   As already mentioned, we set the minimum number of common pairs between $M_1$ and the target matching~$M_2$ (which shall be stable in $P_2$) to be $2h$.
   This completes the construction.

   The following figure depicts a portion of acceptability graph of profile~$P_1$, where the labels on the edges denote the ranks of the agents.
   The edges marked with gray colors are part of  matching~$M_1$.

   
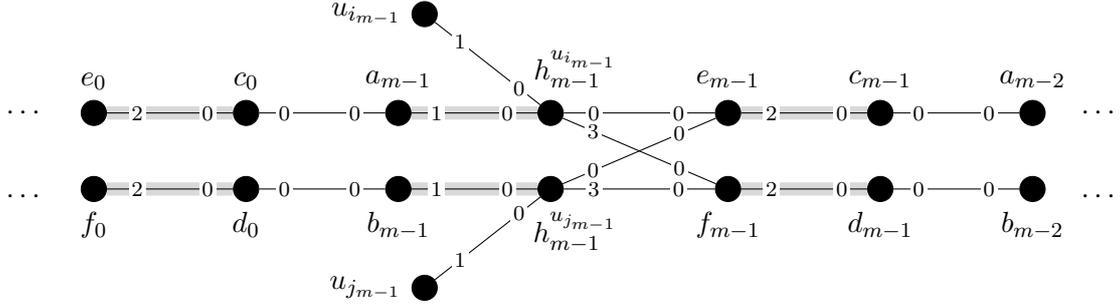
\begin{tikzpicture}
     \node[agent] (e0) {};
     \node[above=0pt of e0]  {$e_0$};
     \node[agent, below = 4ex of e0] (f0) {};
     \node[below = 0pt of f0] {$f_0$};

     \node[agent, right = 10ex of e0] (c0) {};
     \node[above=0pt of c0]  {$c_0$};
     \node[agent, below = 4ex of c0] (d0) {};
     \node[below = 0pt of d0] {$d_0$};

     \node[agent, right = 10ex of c0] (am1) {};
     \node[above = 0pt of am1] {$a_{m-1}$};

     \node[agent, below = 4ex of am1] (bm1) {};
     \node[below = 0pt of bm1] {$b_{m-1}$};

     \node[agent, right = 10ex of am1] (hm11) {};
     \node[above = 0pt of hm11,xshift=2ex] (hm11t) {$h_{m-1}^{u_{i_{m-1}}}$};

     \node[agent, above = 1ex of hm11t, xshift=-12ex] (uim1) {};
     \node[left = 0pt of uim1] {$u_{i_{m-1}}$};

     \node[agent, below = 4ex of hm11] (hm12) {};
     \node[below = 0pt of hm12,xshift=2ex] (hm12t) {$h_{m-1}^{u_{j_{m-1}}}$};

     \node[agent, below = 1ex of hm12t, xshift=-12ex] (ujm1) {};
     \node[left = 0pt of ujm1] {$u_{j_{m-1}}$};

     \node[agent, right = 12ex of hm11] (em1) {};
     \node[above = 0pt of em1] {$e_{m-1}$};

     \node[agent, below = 4ex of em1] (fm1) {};
     \node[below = 0pt of fm1] {$f_{m-1}$};

     \node[agent, right = 10ex of em1] (cm1) {};
     \node[above=0pt of cm1]  {$c_{m-1}$};
     \node[agent, below = 4ex of cm1] (dm1) {};
     \node[below = 0pt of dm1] {$d_{m-1}$};

     \node[agent, right = 10ex of cm1] (am2) {};
     \node[above = 0pt of am2] {$a_{m-2}$};

     \node[agent, below = 4ex of am2] (bm2) {};
     \node[below = 0pt of bm2] {$b_{m-2}$};

     \node[right = 2ex of am2]  (rd1) {$\cdots$};
     \node[below = 4ex of rd1] (rd2) {$\cdots$};

     \node[left = 2ex of e0]  (ld1) {$\cdots$};
     \node[below = 4ex of ld1] (ld2) {$\cdots$};

     \foreach \i / \j / \w / \u in {0/m1/0/0, m1/m2/0/0} {
       \foreach \c/\d in {c/a,d/b} {
         \draw (\c\i) edge node[pos=0.2, fill=white, inner sep=1pt] {\scriptsize $\w$}  node[pos=0.76, fill=white, inner sep=1pt] {\scriptsize $\u$} (\d\j);
     }
   }

     \foreach \i / \j / \w / \u in {0/0/0/2, m1/m1/0/2} {
       \foreach \c/\d in {c/e,d/f} {
         \draw (\c\i) edge node[pos=0.2, fill=white, inner sep=1pt] {\scriptsize $\w$}  node[pos=0.76, fill=white, inner sep=1pt] {\scriptsize $\u$} (\d\j);
     }
   }

   \foreach \c/\d/\u/\w in {am1/hm11/0/1,bm1/hm12/0/1, hm11/uim1/1/0,hm12/ujm1/1/0} {
     \draw (\c) edge node[pos=0.2, fill=white, inner sep=1pt] {\scriptsize $\w$}  node[pos=0.76, fill=white, inner sep=1pt] {\scriptsize $\u$} (\d);
   }

   \foreach \j in {hm11, hm12} {
     \draw (\j) edge node[pos=0.2, fill=white, inner sep=1pt] {\scriptsize $0$}  node[pos=0.76, fill=white, inner sep=1pt] {\scriptsize $0$} (em1);
     \draw (\j) edge node[pos=0.2, fill=white, inner sep=1pt] {\scriptsize $3$}  node[pos=0.76, fill=white, inner sep=1pt] {\scriptsize $0$} (fm1);
   }

   \begin{pgfonlayer}{background}
     \foreach \i / \j in {e0/c0, f0/d0, am1/hm11,bm1/hm12, em1/cm1,fm1/dm1} {
       \draw[line width=5pt, gray!30] (\i) edge  (\j);
     }
   \end{pgfonlayer}

\end{tikzpicture}

   Before we show that this is parameterized reduction, we observe the following.

   \begin{claim}\label{claim:properties:M2}
     Every stable matching~$M$ of $P_2$ satisfies the following.
     \begin{enumerate}
       \item For each vertex~$v_i \in V(G)$, we have that $\{M(u_i), M(w_i)\}=\{x_i,y_i\}$.
       \item For each edge~$e_{\ell}\in E(G)$, agents~$c_{\ell}$ and $d_{\ell}$ must be assigned some partners under $M$.
       \item For each edge~$e_\ell\in E(G)$ with $e_\ell=\{v_i,v_j\}$,
       we have that $M(c_\ell)=a_{\ell-1}$, $M(d_\ell)=b_{\ell-1}$, and $\{M(h^{u_i}_{\ell-1}), M(h_{\ell-1}^{u_j})\} = \{e_{\ell-1}, f_{\ell-1}\}$.
    \item For each vertex~$v_i\in V(G)$ and each edge~$e_{\ell}\in E(G)$ with $v_i \in e_{\ell}$,
       if $M(u_i)=y_i$, then $M(h^{u_i}_{\ell})=e_{\ell}$.
    \end{enumerate}
   \end{claim}
   \begin{proof}\renewcommand{\qedsymbol}{(of \cref{claim:properties:M2})~$\diamond$}
     For the first statement, suppose, towards a contradiction, that there exists some vertex~$v_i$ such that $\{M(u_i), M(w_i)\} \neq \{x_i, y_i\}$.
     By the preference lists of $x_i$ and $y_i$, it follows that $x_i$ or $y_i$ is not assigned a partner by $M$.
     If agent~$x_i$ was not assigned a partner, then it will form with $w_i$ a blocking pair of $M$ as $x_i$ is the only agent which $w_i$ prefers most.
     Analogously, if agent~$y_i$ was not assigned a partner, then it will form with $u_i$ a blocking pair of $M$.

     For the second statement, observe that $a_{\ell-1}$ ranks $c_{\ell}$ at the first place.
     If agent~$c_{\ell-1}$ would have been unmatched under~$M$, then it would form with $a_{\ell-1}$ a blocking pair.
     Analogously, we can obtain that $d_{\ell-1}$ must also be assigned a partner.

     For the third statement, observe that in $P_2$, agents~$c_0$ and $a_{m-1}$ form a fixed pair, meaning that they prefer each other more than any other agents.
     Thus, agents~$c_0$ and $a_{m-1}$ must be matched with each other in any stable matching of $P_2$.
     Analogously, agents~$d_0$ and $b_{m-1}$ must be matched with each other in any stable matching of $P_2$.

     Now, let $v_{i_{m-1}}$ and $v_{j_{m-1}}$ be the two endpoints of edge~$e_{m-1}$. 
     By the preference lists of $e_{m-1}$ and $f_{m-1}$, we know that $h_{m-1}^{u_{i_{m-1}}}$ and $h_{m-1}^{u_{j_{m-1}}}$ must be assigned some partners as otherwise the one that is unmatched (under $M$) will form with $e_{m-1}$ and $f_{m-1}$ two blocking pairs of $M$.
     By the first statement we have that $M(u_{i_{m-1}})\in \{x_{i_{m-1}}, y_{i_{m-1}}\}$ and since $a_{m-1}$ is already matched to $c_0$, agent~$h_{m-1}^{u_{i_{m-1}}}$ can only be matched to an agent that is either $e_{m-1}$ or $f_{m-1}$.
     Analogously, agent~$h_{m-1}^{u_{j_{m-1}}}$ can only be matched to an agent that is either $e_{m-1}$ or $f_{m-1}$.

     We have just reasoned that the partners of $e_{m-1}$ and $f_{m-1}$ under $M$ are from~$\{h_{m-1}^{u_{i_{m-1}}}, h_{m-1}^{u_{j_{m-1}}}\}$.
     By the preference lists of $c_{m-1}$ and $d_{m-1}$, using the second statement,
     we conclude that $M(c_{m-1})=a_{m-2}$ and $M(d_{m-1})=b_{m-2}$.
     By a similar reasoning as we did for $a_{m-1}$ and $b_{m-1}$, we can achieve our desired third statement.

     For the fourth statement, assume that $M(u_i)=y_i$. 
     Since $u_i$ prefers $h_{\ell}^{u_i}$ to its partner~$y_i$,
     by the stability of $M$,
     it follows that $h_{\ell}^{u_i}$ must find its partner~$M(h_{\ell}^{u_i})$ at least as good as $u_i$.
     This means that $M(h_{\ell}^{u_i})\neq f_{\ell}$.
     By the third statement, it follows that $M(h_{\ell}^{u_i})=e_{\ell}$.
   \end{proof}

   Now, we are ready to show the correctness of the construction, i.e.\
   $I$ admits an independent set of size~$h$ if and only if $P_2$ admits a stable matching~$M_2$ with $|M_1\cap M_2|\ge 2h$.
   For the ``only if'' direction, assume that $V'$ is an independent set of $G$ with $h$ vertices.
   We construct matching~$M_2$ as follows.
   \begin{enumerate}
     \item For each vertex~$v_i\in V$, if $v_i\notin V'$ is \emph{not} from the independent set,
     then let $M_2(u_i)=x_i$ and $M_2(w_i)=y_i$; otherwise let $M_2(u_i)=M_1(u_i)=y_i$ and $M_2(w_i)=M_1(w_i)=x_i $.
     \item For each edge~$e_{\ell}\in E(G)$, let $v_i$ and $v_j$ denote the endpoints of $e_{\ell}$ with $i< j$.
     If $v_j\in V'$ is from the independent set, implying that $v_i \notin V'$, then let $M_2(e_{\ell})=h_\ell^{u_j}$ and $M_2(f_{\ell})=h_{\ell}^{u_i}$.
     Otherwise, let $M_2(e_{\ell})=h_\ell^{u_i}$ and $M_2(f_{\ell})=h_{\ell}^{u_j}$.
     Let $M_2(a_\ell)=c_{\ell+1}$ and $M_2(b_\ell)=d_{\ell+1}$.
   \end{enumerate}

   It is straight-forward to verify that $|M_1\cap M_2| = 2h$ as they share at the pairs that correspond to the vertices in the independent set.
   Now, we focus on the stability of $M_2$.
   Towards a contradiction, suppose that $M_2$ is not stable in $P_2$, and let $p$ be a blocking pair of $M_2$.
   This pair~$p$ must involve some agent from $U\cup W\cup E \cup F \cup A \cup B$.

   First of all, one can verify that no agent~$w_i$ from $W$ would be involved in a blocking pair because $x_i$ is the only agent with which $w_i$ could form a blocking pair, but $x_i$ will already obtain its most preferred agent~$y_i$.
   Further, no agent from $E\cup F \cup A \cup B$ would be involved in a blocking pair as they already received one their most preferred agents.
   Thus, we obtain that $p$ involves some agent from $U$, say $u_z$.
   We know that $M_2(u_z)=y_z$ as otherwise we have that $M_2(u_z)=x_z$ which is $u_z$'s most preferred agent--a contradiction to $u_z$ being in blocking pair~$p$.
   By our definition of $M_2$, it follows that $v_z\in V'$ is from the independent set.
   Hence, the other agent in blocking pair~$p$ must be some agent~$h_{\ell}^{u_z}$ with $v_z \in e_{\ell}$ such that
   $h_{\ell}^{u_z}$ prefers $u_z$ to $M_2(h_{\ell}^{u_z})$.
   By the preference list of $h_{\ell}^{u_z}$ it follows that $M_2(h_{\ell}^{u_z})=f_\ell$ as otherwise $h_{\ell}^{u_z}$ will not form with $u_z$ a blocking pair.
   Let $v_{z'}$ be the other endpoint of $e_{\ell}$, implying that $v_{z'}\notin V'$.
   If $z< z'$, then by our definition of $M_2$, it follows that $M_2(h_{\ell}^{u_z})=e_{\ell}$--a contradiction.
   If $z' < z$, then by our definition of $M_2$, it also follows that $M_2(h_{\ell}^{u_z})=e_{\ell}$--a contradiction.
This shows that $M_2$ is indeed a stable matching of $P_2$ which shares with $M_1$ by $2h$~pairs.

For the other direction, assume that $M_2$ is a stable matching of $P_2$, sharing with $M_1$ by at least $2h$~pairs.
By the third statement in \cref{claim:properties:M2}, it follows that $M_2$ can only share with $M_1$ by the pairs that involve some agents from $U\cup W\cup X \cup Y$.
Now we construct a vertex subset according to these pairs.
Let $V'=\{v_i \in V\mid M_2(u_i)=y_i\}$.
Clearly, using the first statement in \cref{claim:properties:M2}, we have that $|V'|\ge h$ as $|M_1\cap M_2|\ge 2h$.
Now, to show that $V'$ is indeed an independent set, suppose, towards a contradiction,
that $V'$ has two adjacent vertices, denoted as $v_i$ and $v_j$.
Let $e_{\ell}$ be the incident edge of $v_i$ and $v_j$.
By the fourth statement of \cref{claim:properties:M2} we have that $M_2(h_{\ell}^{u_i})=e_\ell=M_2(h_{\ell}^{u_j})$--a contradiction to $M_2$ being a matching.
\end{proof}
}

Finally, we show that even a single swap in the preference list of one single agent makes the \ISR problem \woneh{} with respect to~$k$ when ties are allowed.
To show this result, we give a reduction from \EIIS.
   The construction idea is inspired by a reduction from \textsc{Vertex Cover}
   to \SR with structured preferences~\cite{BCFN17}, which, however,
   is relying on incomplete preferences and not showing parameterized intractability.

\begin{theorem}
  \label{thm:W1h-ISR-ties-ONESWAP}
  \ISR with ties and complete preferences is \woneh{} with respect to~$k=|M_1 \Delta M_2|$,
  even if\/ $\left| \mathcal{P}_1 \oplus \mathcal{P}_2 \right|=1$.
  It remains \np-hard even if\/ $\left| \mathcal{P}_1 \oplus \mathcal{P}_2 \right|=1$,
  and the sought stable matching~$M_2$ only needs to satisfy that $\left| M_1 \cap M_2 \right| \ge 0$.
\end{theorem}

\newcommand*\showwidth[1]{%
  \textcolor{blue}{\rule{\csname#1\endcsname}{1pt}}\newline
  \texttt{\textbackslash#1}: \expandafter\the\csname#1\endcsname
  \par
}

\appendixproof{thm:W1h-ISR-ties-ONESWAP}
{
  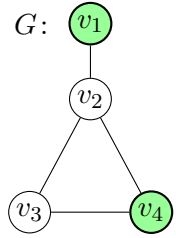
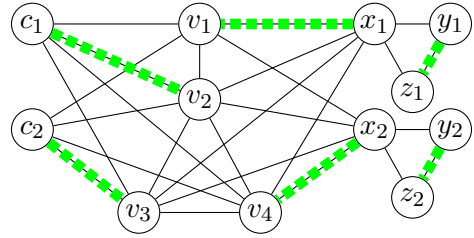
\begin{figure}[t]
    \begin{subfigure}[t]{0.45\textwidth}
      \centering
      \begin{tikzpicture}
      [scale=1.0,auto]
      \node[nodest] (n1) at (0,0) [thick,fill=green!40] {$v_1$};
      \node[nodest] (n2) at (0,-1) {$v_2$};
      \node[nodest] (n3) at (-.8,-2.5) {$v_3$};
      \node[nodest] (n4) at (.8,-2.5) [thick,fill=green!40] {$v_4$};
      \node[left = 0pt of n1] {$G\colon$};

      \foreach \from/\to in {n1/n2,n2/n3,n2/n4,n3/n4}
      \draw (\from) -- (\to);
      \end{tikzpicture}
      \caption{The graph of an \EIIS{} instance~$(G,e^*=\{v_3, v_4\}, h=2, S^*=\{v_3, v_4\})$.
              The instance is a yes-instance: The graph~$G$ admits an independent set~$\{v_1,v_4\}$ of size two, marked in light green.} \label{fig:W1h-ISR-ties-ONESWAP A}
    \end{subfigure}
    \hspace{0.55cm}
    \begin{subfigure}[t]{0.45\textwidth}
      \centering
      \begin{tikzpicture}
      [scale=1.0,auto]
        \node[nodest] (n1) at (0,0) {$v_1$};
        \node[nodest] (n2) at (0,-1) {$v_2$};
        \node[nodest] (n3) at (-0.8,-2.5) {$v_3$};
        \node[nodest] (n4) at (0.8,-2.5) {$v_4$};

        \node[nodest] (t1) at (2.3,-1.4) {$x_2$};
        \node[nodest] (t11) at (3.3,-1.4) {$y_2$};
        \node[nodest] (t12) at (2.8,-2.3) {$z_2$};

        \node[nodest] (t2) at (2.3,-0) {$x_1$};
        \node[nodest] (t21) at (3.3,0) {$y_1$};
        \node[nodest] (t22) at (2.8,-.9) {$z_1$};

        \node[nodest] (s1) at (-2.2,0) {$c_1$};
        \node[nodest] (s2) at (-2.2,-1.4) {$c_2$};

        \foreach \from/\to in {n1/n2,n2/n3,n2/n4,n3/n4}
          \draw (\from) -- (\to);
        \foreach \from/\to in {s1/n1,s1/n3,s1/n4}
          \draw (\from) -- (\to);
        \foreach \from/\to in {s2/n1,s2/n2,s2/n4}
          \draw (\from) -- (\to);
        \foreach \from/\to in {t1/n1,t1/n2,t1/n3,t1/t11,t1/t12}
          \draw (\from) -- (\to);
        \foreach \from/\to in {t2/n2,t2/n3,t2/n4,t2/t21,t2/t22}
          \draw (\from) -- (\to);
        \foreach \from/\to in {s1/n2,s2/n3,t1/n4,t2/n1,t21/t22,t11/t12}
          \draw (\from) -- (\to);
        \foreach \from/\to in {s1/n2,s2/n3,t1/n4,t2/n1,t21/t22,t11/t12}
          \draw[green,line width=4pt,dotted] (\from) -- (\to);
      \end{tikzpicture}
      \caption{The acceptability graph of the corresponding \ISR instance (for both $\mathcal{P}_1$ and $\mathcal{P}_2$). Profile~$\mathcal{P}_2$ admits a stable matching, marked by the thick dotted green lines.}\label{fig:W1h-ISR-ties-ONESWAP B}
    \end{subfigure}
    \caption{An illustration of the hardness reduction for \cref{thm:W1h-ISR-ties-ONESWAP}.}\label{fig:W1h-ISR-ties-ONESWAP}
  \end{figure}

\begin{proof}
  Let $(G=(V,E),e^* \in E, h, S^* \subseteq V)$ be an \EIIS instance with $r \coloneqq |V|$ vertices.
  We assume w.l.o.g.\ that $h < r$, $S^*=\{v_{r-h+1},\dots,v_r\}$.
  Moreover, we assume that $e^*\subseteq S^*$ (as otherwise the input instance is a trivial yes-instance) with $e^*=\{v_{r-1}, v_r\}$.
  We construct an \ISR instance~$(\mathcal{P}_1,\mathcal{P}_2,M_1)$ with agent set~$U$
  and $\left| \mathcal{P}_1 \oplus \mathcal{P}_2 \right|=1$.
  We will show that~$G$ has an independent of size at least~$h$
  if and only if~$\mathcal{P}_2$ admits a stable matching~$M_2$ with $\left| M_1 \oplus M_2 \right|\le 4h$.

   Before we describe the construction, we prove the following claim which is heavily used
   in our preference profile construction to force two agents to be matched together.

   \begin{claim} \label{cl:TopTGadget}
    Let $\mathcal{P}$~be a profile for an agent set~$U$,
    and let $x$, $y$, and $z$ be three distinct agents with the following preference lists, 

    $\begin{array}{rcccccccrccccccc}
     \text{agent}~x\colon &V&\succ&y&\succ&z&\succ& \ldots, \\
     \text{agent}~y\colon &z&\succ&x&\succ& \ldots, \\
     \text{agent}~z\colon &x&\succ&y&\succ& \ldots, 
    \end{array}$

    \noindent  where $V\subseteq U \setminus \{x, y, z\}$ is a non-empty subset of agents,
    The symbol ``$\ldots$'' at the preference list of each agent~$a$ denotes  an arbitrary but fixed order of all remaining agents, other than $a$ and not explicitly stated before~``$\ldots$'':
    Then, every stable matching~$M$ for~$\mathcal{P}$ must fulfill that
    (i) $M(x)\in V$ and (ii) $\{y,z\}\in M$.
   \end{claim}

   \begin{proof}\renewcommand{\qedsymbol}{(of \cref{cl:TopTGadget})~$\diamond$}
    Assume towards a contradiction to (i) that $\mathcal{P}$ admits a stable matching~$M$ with $M(x)\notin V$.
    There are three cases:
    (1) $M(x)=y$, implying the blocking pair~$\{y,z\}$,
    (2) $M(x)=z$, implying the blocking pair~$\{x,y\}$, and
    (3) $M(x)\notin \{y,z\} \cup V$, implying the blocking pair~$\{x,z\}$.
    Thus, $x$~must be matched with some agent from~$V$.
    For (ii),  statement~(i) implies that neither $y$ nor $z$ is matched with~$x$.
    Now, if $\{y,z\}\notin M$, then~$\{y,z\}$ forms a blocking pair.
   \end{proof}

  \paragraph{Main idea and the constructed agents.}
  To explain the main idea of the reduction, we first describe the agent set~$U$ and
  the corresponding non-complete acceptability graph of~$\mathcal{P}$ as illustrated through an example in \cref{fig:W1h-ISR-ties-ONESWAP}.
  At the end of the proof, we show that we can adjust the reduction
  so that the acceptability graph becomes complete.

  For each vertex~$v_i \in V$, we introduce a \emph{vertex agent}~$v_i$
  (for the sake of simplicity, we use the same symbol for the vertex and the corresponding agent).
  Additionally, there is a set of \emph{cover agents}~$C\coloneqq \{c_1,\dots,c_{r-h}\}$ as well as
  three sets of \emph{selector agents} $X\coloneqq \{x_1,x_2,\ldots,x_{h}\}$,
  $Y\coloneqq \{y_1,y_2,\ldots,y_{h}\}$, and $Z\coloneqq \{z_1,z_2,\ldots,z_{h}\}$.
  The agent set~$U$ is defined as $V \cup C \cup X\cup Y\cup Z$.
  For the acceptability graph,
  we have that every vertex agent~$v_i$ accepts every cover agent from~$C$,
  every selector agent from~$X$,
  and every vertex agent~$v_j$ that corresponds to a neighbor of $v_i$ in the input graph~$G$.
  For each $i \in \{1,2,\ldots, h\}$,
  the collector agents $x_i$, $y_i$, and $z_i$ pairwisely accept each other.

  Using \cref{cl:TopTGadget} we construct the preference profile~$\mathcal{P}_2$ of the agents such that
  in every stable matching only the cover agents from~$C$ and the selector agents from~$X$ can be matched to the
  vertex agents and the vertex agents matched to the selector agents from~$X$ correspond to an independent set (of size~$|X|=h$).
  These properties are given by the subsequent \cref{claim:stable-matching-properties}.

  \paragraph{Agent preferences in $\mathcal{P}_1$ and $\mathcal{P}_2$.}
  First, we describe the preferences profiles $\mathcal{P}_2$ that realize the idea and the acceptability graph as described above.
  \begin{alignat*}{3}
   \mathcal{P}_2 \text{: }& \agent~v_i\colon & &[C] \succ [N(v_i)]  \succ x_1 \succ x_2 \succ \cdots \succ x_{h} \succ \ldots, &\quad& \forall 1\le i \le r-1,\\
  &  \agent~v_r\colon && [C] \succ [N(v_i)\setminus\{v_{r-1}\}] \succ v_{r-1}  \succ x_1 \succ x_2 \succ \cdots \succ x_{h}\succ \ldots, &\quad&\\
   & \agent~c_i\colon &&  v_1 \sim v_2 \sim \cdots \sim v_r\succ \ldots, & \quad& \forall 1\le i \le {r-h},\\
   & \agent~x_i\colon & & v_1 \sim v_2 \sim \cdots \sim v_r \succ y_i \succ z_i\succ \ldots, & \quad& \forall 1\le i \le h,\\
   & \agent~y_i\colon & & z_i \succ x_i\succ \ldots, & \quad& \forall 1\le i \le h,\\
   & \agent~z_i\colon &&  x_i \succ y_i\succ \ldots, & \quad& \forall 1\le i \le h.
  \end{alignat*}

  With a single swap in the preference list of agent~$v_r$, we obtain the profile~$\mathcal{P}_1$:
  \begin{alignat*}{3}
    \mathcal{P}_1 \text{: }& \agent~v_i\colon && [C] \succ [N(v_i)]  \succ x_1 \succ x_2 \succ \ldots \succ x_{h}\succ \ldots, &\quad& \forall 1\le i \le r-1,\\
    &  \agent~v_r\colon && [C] \succ [N(v_i)\setminus\{v_{r-1}\}] \succ x_1  \succ v_{r-1} \succ x_2 \succ \cdots \succ x_{h}\succ \ldots, &\quad&\\
    &  \agent~c_i\colon & & v_1 \sim v_2 \sim \cdots \sim v_r\succ \ldots, & \quad& \forall 1\le i \le {r-h},\\
    &  \agent~x_i\colon & &  v_1 \sim v_2 \sim \cdots \sim v_r \succ y_i \succ z_i\succ \ldots, & \quad& \forall 1\le i \le h,\\
    & \agent~y_i\colon & & z_i \succ x_i\succ \ldots, & \quad& \forall 1\le i \le h,\\
    &  \agent~z_i\colon & & x_i \succ y_i\succ \ldots, & \quad& \forall 1\le i \le h.
  \end{alignat*}

  \noindent Herein, for each agent subset~$D$,
  we denote by $[D]$ the lexicographic order of the agents by their names (resp.\ by their indices), called the \emph{canonical order}.
  The symbol~``$\ldots$'' at the preference list of each agent~$a$ denotes  an arbitrary but fixed order of all remaining agents, other than $a$ and not explicitly stated before~``$\ldots$'':

  \paragraph{Stable matching~$M_1$ for $\mathcal{P}_1$.}
  The initial matching~$M_1$ is defined as follows:
  \begin{itemize}
    \item for each $i\in \{1,2,\ldots, h\}$ set $M(v_{r-i+1})\coloneqq x_i$;
    \item for each $i\in \{1,2,\ldots, r-h\}$
    set $M(v_{i})=c_i$;
    \item for each $i\in \{1,2,\ldots, h\}$ set $M(y_i) = z_i$.
  \end{itemize}

  This completes the construction and can clearly be performed in polynomial~time.

  \paragraph{Correctness of the construction.}
  First of all, we claim that in the constructed \ISR instance~$M_1$ is indeed a stable matching for $\mathcal{P}_1$.
  \begin{claim}\label{claim:GivenM1stable}
    The matching~$M_1$ is a stable matching for the profile $\mathcal{P}_1$.
  \end{claim}
  \begin{proof}[Proof of \cref{claim:GivenM1stable}]\renewcommand{\qed}{\hfill~(of \cref{claim:GivenM1stable})~$\diamond$}
   First, observe that, by definition of~$M_1$, no agent from~$C$ or~$X$ can be involved in a blocking pair,
   since they all are matched to one of their (tied) favorite partners.
   Second, since the agents from~$X$ are not available also the agents from~$Y$ and from~$Z$ cannot be involved in blocking pairs.
   Finally, a blocking pair that consist of two agents from~$V$ (the last possibility for a blocking pair to emerge) would
   mean that we have two agents~$v_{r-i+1}i$ and~$v_{r-j+1}$ ($1\le i\neq j \le h$) that are both matched to some agent from $X$.
   Furthermore, in order to be a blocking pair both agents must be adjacent in~$G$ because only agents corresponding
   to neighbors in~$G$ are preferred to agents from~$X$.
   Since the vertices $v_{r-h+1},\dots,v_r$ are an independent set in $G - e^*$ for $e^*=\{v_{r-1},v_r\}$,
   only $\{v_{r-1},v_r\}$ could be a blocking pair.
   However, $\{v_{r-1},v_r\}$ is also not blocking $M_1$ because
   (due to the swap of~$x_1$ and~$v_{r-1}$ in the preferences of~$v_r$ between~$\mathcal{P}_2$ and~$\mathcal{P}_1$)
   agent~$v_r$ prefers its partner $M_1(v_r)=x_1$ over~$v_{r-1}$.
  \end{proof}

  Second,  we claim the following:
  \begin{claim}\label{claim:stable-matching-properties}
    Every stable matching~$M'_2$ for $\mathcal{P}_2$
    satisfies the following two properties:
    \begin{enumerate}
      \item \label{prop:vertices-matched}
      every vertex agent~$v_i$ is matched to either a cover agent from~$C$ or a selector agent from~$X$, that is,
      $M(v_i)\in C\cup X$ for every $v \in V$, and
      \item \label{prop:edges-covered}
      \emph{no} two vertex agents that are both matched to the selector agents are adjacent in~$G$.
    \end{enumerate}
  \end{claim}

  \begin{proof}[Proof of \cref{claim:stable-matching-properties}]
    \renewcommand{\qed}{\hfill~(of \cref{claim:stable-matching-properties})~$\diamond$}
    Let $M'_2$~be a stable matching for~$\mathcal{P}$.
    For the first statement, observe \cref{cl:TopTGadget} immediately implies that
    for every selector agent~$x_i \in X$, it holds that $M(x_i) \in V$.
    Thus, there are exactly $r-h$~vertex agents left that are not matched to agents from~$X$.
    Suppose towards a contradiction that
    some cover agent~$c_j$ is not matched to any vertex agent,
    implying that at least one vertex agent~$v_i$ is left with $M(v_i)\notin X\cup C$.
    This, however, implies that $\{c_j, v_i\}$ is a blocking pair for $M'_2$---a contradiction.
    For the second statement,
    suppose towards a contradiction that there are two vertex agents~$v_i, v_j$ with $\{M(v_i), M(v_j)\}\subseteq X$
    as well as $\{v_i,v_j\}\in E$.
    The preference lists of $v_i$ and $v_j$ immediately imply that agents $v_i$ and $v_j$ form a blocking pair---a contradiction.
  \end{proof}

  Now, we show that~$G$ has an independent set of size at least~$h$ if and only if $\mathcal{P}_2$ admits a stable matching~$M_2$
  with $\left| M_1 \oplus M_2 \right|\le 4h$.

  The ``if'' part follows immediately from \cref{claim:stable-matching-properties} because
  it directly implies that the vertices corresponding to the $h$~vertex agents matched to selector agents from~$X$
  are pairwisely non-adjacent in the graph~$G$.

  For the ``only if'' part, assume that $S\subseteq V$ be an independent set of size~$h$ for~$G$.
  Let~$S\coloneqq \{v_{s_1},v_{s_2},\ldots, v_{s_h}\}$ with $s_1 < s_2 < \cdots < s_h$ 
  such that the canonical order~$\succ$ implies $v_{s_1}\succ v_{s_2} \succ \cdots \succ v_{s_h}$.
  We construct a matching~$M_2$ (based on $M_1$) such that $\left| M_1 \oplus M_2 \right|\le 4h$
  and such that $M_2$ is stable for $\mathcal{P}_2$.
  First, we match all selector agents and the vertex agents corresponding to the independent set~$S$:
  \begin{itemize}
    \item for each $i\in \{1,2,\ldots, h\}$ set $M_2(y_i) \coloneqq z_i (= M_1(y_i))$;
    \item for each $i\in \{1,2,\ldots, h\}$ set $M_2(v_{s_i})\coloneqq x_i$.
  \end{itemize}
  So far, at most $2h$~agents have different partners in~$M_1$ and~$M_2$.
  At least $r-2h$~vertices are neither in the independent set~$S$ for~$G$
  nor in the independent set $S^*=\{v_{r-h+1},\dots,v_r\}$ for $G - e^*$.
  We keep the partners for the corresponding vertex agents unchanged:
  \begin{itemize}
    \item for each $v \in V \setminus (S \cup S^*)$ set $M_2(v)\coloneqq M_1(v)$.
  \end{itemize}
  Finally, some number~$q \le h$ of vertex agents, namely those corresponding to vertices from $S^* \setminus S$,
  are so far unmatched in~$M_2$.
  Symmetrically, also $q \le h$ cover agents, namely those cover agents that where matched to vertex agents from $S^* \setminus S$ in $M_1$,
  are also so far unmatched in~$M_2$.
  We pair these agents arbitrarily such that each vertex agent is matched to a cover agent to finalize the definition of~$M_2$ so far.
  Doing this, at most $2q \le 2h$~agents will get different partners in~$M_2$ when compared with~$M_1$.
  Hence, $|M_1 \Delta M_2| \le 4h$.
  
  It remains to show that~$M_2$ is stable.
  This can be done analogously to the proof of \cref{claim:GivenM1stable}.
  First, observe that, by the definition of~$M_2$, no agent from~$C$ or~$X$ can be involved in a blocking pair
  since they all are matched to one of their (tied) favorite partners.
  Second, since agents from~$X$ are not available, the agents from~$Y$ and from~$Z$ cannot be involved in blocking pairs.
  Finally, a blocking pair that consist of two agents from~$V$ (the last possibility for a blocking pair to emerge) would
  mean that we have two agents~$v_i$ and~$v_j$ that are both matched to some agent from $X$.
  Furthermore, in order to be a blocking pair both agents must be adjacent in~$G$ because only agents corresponding
  to neighbors in~$G$ are preferred to agents from~$X$.
  This is a contradiction to our definition of~$M_2$, where the $h$~agents from~$X$ have been matched to $h$~vertex agents
  corresponding to the independent set~$S$.


  \paragraph{Unbounded $k$.}
  Note that in this proof, we never used an explicit bound on the parameter~$k$ for any
  proof argument to work.
  In particular, if there is any stable matching for~$\mathcal{P}_2$, then there is one
  with distance $\left| M_1 \oplus M_2 \right|\le 4h$
  (this follows from \cref{claim:stable-matching-properties} and from the idea given in the ``only if'' part).
  In particular, the problem remains \np-hard if we set~$k=2n$, that is, $k'=2n-k=0$.
  In other words, the problem remains hard even if the distance between the matchings~$M_1$
  and~$M_2$ is unbounded.
\end{proof}
}

\section{Conclusion}
\label{sec:conclusion}
Motivated by dynamically changing preferences and the necessity
to adapt the corresponding solutions, we introduced an ``incremental view''
on the computation of stable matchings. We believe that there are plenty of opportunities for future research, including, for instance, to study the
role of parameters measuring the number of ties---in many hardness reductions
this parameter is unbounded. We also left open whether \textsc{Incremental
Stable Roommates} without ties is fixed-parameter tractable
when parameterized by the swap distance between the two input preference profiles.
Naturally, there are also future directions concerning more conceptual
work, e.g., also studying further stability concepts in the context of
our incremental model.

%
%
%
%
%
%
%




\begin{thebibliography}{33}
\providecommand{\natexlab}[1]{#1}
\providecommand{\url}[1]{\texttt{#1}}
\expandafter\ifx\csname urlstyle\endcsname\relax
  \providecommand{\doi}[1]{doi: #1}\else
  \providecommand{\doi}{doi: \begingroup \urlstyle{rm}\Url}\fi

\bibitem[Abu-Khzam et~al.(2015)Abu-Khzam, Egan, Fellows, Rosamond, and
  Shaw]{Abu-Khzam+2015}
F.~N. Abu-Khzam, J.~Egan, M.~R. Fellows, F.~A. Rosamond, and P.~Shaw.
\newblock On the parameterized complexity of dynamic problems.
\newblock \emph{Theoretical Computer Science}, 607:\penalty0 426--434, 2015.

\bibitem[Aziz et~al.(2016)Aziz, Bir{\'{o}}, Gaspers, de~Haan, Mattei, and
  Rastegari]{AzBiGaHaMaRa2016}
H.~Aziz, P.~Bir{\'{o}}, S.~Gaspers, R.~de~Haan, N.~Mattei, and B.~Rastegari.
\newblock Stable matching with uncertain linear preferences.
\newblock In \emph{Proceedings of the 9th International Symposium on
  Algorithmic Game Theory}, pages 195--206, 2016.

\bibitem[Bhattacharya et~al.(2015)Bhattacharya, Hoefer, Huang, Kavitha, and
  Wagner]{BHHKW15}
S.~Bhattacharya, M.~Hoefer, C.~Huang, T.~Kavitha, and L.~Wagner.
\newblock Maintaining near-popular matchings.
\newblock In \emph{Proceedings of the 42nd International Colloquium on
  Automata, Languages, and Programming}, pages 504--515, 2015.

\bibitem[B{\"{o}}ckenhauer et~al.(2018)B{\"{o}}ckenhauer, Burjons, Raszyk, and
  Rossmanith]{abs-1809-10578}
H.~B{\"{o}}ckenhauer, E.~Burjons, M.~Raszyk, and P.~Rossmanith.
\newblock Reoptimization of parameterized problems.
\newblock \emph{CoRR}, abs/1809.10578, 2018.

\bibitem[Bredereck et~al.(2017)Bredereck, Chen, Finnendahl, and
  Niedermeier]{BCFN17}
R.~Bredereck, J.~Chen, U.~P. Finnendahl, and R.~Niedermeier.
\newblock Stable roommate with narcissistic, single-peaked, and single-crossing
  preferences.
\newblock In \emph{Proceedings of the 5th International Conference on
  Algorithmic Decision Theory}, pages 315--330, 2017.

\bibitem[Charikar et~al.(2004)Charikar, Chekuri, Feder, and
  Motwani]{ChChFeMo2004}
M.~Charikar, C.~Chekuri, T.~Feder, and R.~Motwani.
\newblock Incremental clustering and dynamic information retrieval.
\newblock \emph{SIAM Journal on Computing}, 33\penalty0 (6):\penalty0
  1417--1440, 2004.

\bibitem[Chen et~al.(2018{\natexlab{a}})Chen, Hermelin, Sorge, and
  Yedidsion]{CHSYicalp-par-stable2018}
J.~Chen, D.~Hermelin, M.~Sorge, and H.~Yedidsion.
\newblock How hard is it to satisfy (almost) all roommates?
\newblock In \emph{Proceedings of the 45th International Colloquium on
  Automata, Languages, and Programming}, pages 35:1--35:15, 2018{\natexlab{a}}.

\bibitem[Chen et~al.(2018{\natexlab{b}})Chen, Niedermeier, and
  Skowron]{CheNieSkoECmstable2018}
J.~Chen, R.~Niedermeier, and P.~Skowron.
\newblock Stable marriage with multi-modal preferences.
\newblock In \emph{Proceedings of the 19th ACM Conference on Economics and
  Computation}, pages 269--286, 2018{\natexlab{b}}.

\bibitem[Chen et~al.(2019)Chen, Skowron, and
  Sorge]{ChenSkowronSorge2019ec-robust}
J.~Chen, P.~Skowron, and M.~Sorge.
\newblock Matchings under preferences: {S}trength of stability and trade-offs.
\newblock In \emph{Proceedings of the 20th ACM Conference on Economics and
  Computation}, pages 41--59, 2019.

\bibitem[Cseh and Manlove(2016)]{CsehManlove2016}
{\'{A}}.~Cseh and D.~F. Manlove.
\newblock Stable marriage and roommates problems with restricted edges:
  Complexity and approximability.
\newblock \emph{Discrete Optimization}, 20:\penalty0 62--89, 2016.

\bibitem[Drummond and Boutilier(2013)]{DruBou13}
J.~Drummond and C.~Boutilier.
\newblock Elicitation and approximately stable matching with partial
  preferences.
\newblock In \emph{Proceedings of the 23rd International Joint Conference on
  Artificial Intelligence}, pages 97--105, 2013.

\bibitem[Feder(1992)]{feder1992new}
T.~Feder.
\newblock A new fixed point approach for stable networks and stable marriages.
\newblock \emph{Journal of Computer and System Sciences}, 45\penalty0
  (2):\penalty0 233--284, 1992.

\bibitem[Genc et~al.(2017{\natexlab{a}})Genc, Siala, O'Sullivan, and
  Simonin]{GenSiaOSuSim2017ijcai}
B.~Genc, M.~Siala, B.~O'Sullivan, and G.~Simonin.
\newblock Finding robust solutions to stable marriage.
\newblock In \emph{Proceedings of the 26th International Joint Conference on
  Artificial Intelligence}, pages 631--637, 2017{\natexlab{a}}.

\bibitem[Genc et~al.(2017{\natexlab{b}})Genc, Siala, Simonin, and
  O'Sullivan]{GenSiaSimOSul17cocoa}
B.~Genc, M.~Siala, G.~Simonin, and B.~O'Sullivan.
\newblock On the complexity of robust stable marriage.
\newblock In \emph{Proc. COCOA-17}, pages 441--448, 2017{\natexlab{b}}.

\bibitem[Genc et~al.(2017{\natexlab{c}})Genc, Siala, Simonin, and
  O'Sullivan]{GenSiaSimSul17aaai}
B.~Genc, M.~Siala, G.~Simonin, and B.~O'Sullivan.
\newblock Robust stable marriage.
\newblock In \emph{Proceedings of the 31st AAAI Conference on Artificial
  Intelligence}, pages 4925--4926, 2017{\natexlab{c}}.

\bibitem[Genc et~al.(2019)Genc, Siala, Simonin, and
  O'Sullivan]{genc_complexity_2019tcs}
B.~Genc, M.~Siala, G.~Simonin, and B.~O'Sullivan.
\newblock Complexity study for the robust stable marriage problem.
\newblock \emph{Theoretical Computer Science}, 775:\penalty0 76--92, 2019.

\bibitem[Ghosal et~al.(2017)Ghosal, Kunysz, and Paluch]{GKP17}
P.~Ghosal, A.~Kunysz, and K.~E. Paluch.
\newblock The dynamics of rank-maximal and popular matchings.
\newblock \emph{CoRR}, abs/1703.10594, 2017.

\bibitem[Gusfield(1988)]{gusfield1988structure}
D.~Gusfield.
\newblock The structure of the stable roommate problem: {E}fficient
  representation and enumeration of all stable assignments.
\newblock \emph{SIAM Journal on Computing}, 17\penalty0 (4):\penalty0 742--769,
  1988.

\bibitem[Gusfield and Irving(1989)]{GusfieldIrving1989}
D.~Gusfield and R.~W. Irving.
\newblock \emph{The Stable Marriage Problem--{S}tructure and Algorithms}.
\newblock Foundations of Computing Series. {MIT} Press, 1989.

\bibitem[Hartung and Niedermeier(2013)]{HarNie2013}
S.~Hartung and R.~Niedermeier.
\newblock Incremental list coloring of graphs, parameterized by conservation.
\newblock \emph{Theoretical Computer Science}, 494:\penalty0 86--98, 2013.

\bibitem[Irving(1985)]{irving1985efficient}
R.~W. Irving.
\newblock An efficient algorithm for the “stable roommates” problem.
\newblock \emph{Journal of Algorithms}, 6\penalty0 (4):\penalty0 577--595,
  1985.

\bibitem[Irving(1994)]{Irving1994}
R.~W. Irving.
\newblock Stable marriage and indifference.
\newblock \emph{Discrete Applied Mathematics}, 48\penalty0 (3):\penalty0
  261--272, 1994.

\bibitem[Kanade et~al.(2016)Kanade, Leonardos, and Magniez]{KanLeoMag2016}
V.~Kanade, N.~Leonardos, and F.~Magniez.
\newblock Stable matching with evolving preferences.
\newblock In \emph{Approximation, Randomization, and Combinatorial
  Optimization. Algorithms and Techniques}, volume~60 of \emph{LIPIcs}, pages
  36:1--36:13, 2016.
\newblock ISBN 978-3-95977-018-7.

\bibitem[Krithika et~al.(2018)Krithika, Sahu, and Tale]{krithika2018dynamic}
R.~Krithika, A.~Sahu, and P.~Tale.
\newblock Dynamic parameterized problems.
\newblock \emph{Algorithmica}, 80\penalty0 (9):\penalty0 2637--2655, 2018.

\bibitem[Luo et~al.(2018)Luo, Molter, Nichterlein, and Niedermeier]{LuoMNN18}
J.~Luo, H.~Molter, A.~Nichterlein, and R.~Niedermeier.
\newblock Parameterized dynamic cluster editing.
\newblock In \emph{Proceedings of the 38th International Conference on
  Foundations of Software Technology and Theoretical Computer Science}, pages
  46:1--46:15, 2018.

\bibitem[Mai and Vazirani(2018{\natexlab{a}})]{MaiVaz2018}
T.~Mai and V.~V. Vazirani.
\newblock Finding stable matchings that are robust to errors in the input.
\newblock In \emph{Proceedings of the 26th Annual European Symposium on
  Algorithms}, pages 60:1--60:11, 2018{\natexlab{a}}.

\bibitem[Mai and Vazirani(2018{\natexlab{b}})]{MaiVaz2018-birkhoff-arxiv}
T.~Mai and V.~V. Vazirani.
\newblock A generalization of {Birkhoff's} theorem for distributive lattices,
  with applications to robust stable matchings.
\newblock Technical report, arXiv:1804.05537 [cs.DM], 2018{\natexlab{b}}.

\bibitem[Manlove(2013)]{Manlove2013}
D.~F. Manlove.
\newblock \emph{Algorithmics of Matching Under Preferences}.
\newblock WorldScientific, 2013.

\bibitem[Marx and Schlotter(2010)]{MarxSchlotter2010}
D.~Marx and I.~Schlotter.
\newblock Parameterized complexity and local search approaches for the stable
  marriage problem with ties.
\newblock \emph{Algorithmica}, 58\penalty0 (1):\penalty0 170--187, 2010.

\bibitem[Marx and Schlotter(2011)]{MarxSchlotter2011}
D.~Marx and I.~Schlotter.
\newblock Stable assignment with couples: {P}arameterized complexity and local
  search.
\newblock \emph{Discrete Optimization}, 8\penalty0 (1):\penalty0 25--40, 2011.

\bibitem[Miyazaki and Okamoto(2017)]{MiOk2017}
S.~Miyazaki and K.~Okamoto.
\newblock Jointly stable matchings.
\newblock In \emph{Proceedings of the 28th International Symposium on
  Algorithms and Computation}, pages 56:1--56:12, 2017.

\bibitem[Nimbhorkar and Rameshwar(2019)]{NV19}
P.~Nimbhorkar and V.~Rameshwar.
\newblock Dynamic rank-maximal and popular matchings.
\newblock \emph{Journal of Combinatorial Optimization}, 37\penalty0
  (2):\penalty0 523--545, 2019.

\bibitem[Schieber et~al.(2018)Schieber, Shachnai, Tamir, and
  Tamir]{schieber2018theory}
B.~Schieber, H.~Shachnai, G.~Tamir, and T.~Tamir.
\newblock A theory and algorithms for combinatorial reoptimization.
\newblock \emph{Algorithmica}, 80\penalty0 (2):\penalty0 576--607, 2018.

\end{thebibliography}

\appendix

\clearpage

\section{\ISM without ties}
\label{sec:ISM without ties}
In this section, we show that \ISM without ties can be solved in polynomial time,
using an idea similar to the one for finding the so-called \emph{maximum weight stable matching}, a stable matching whose corresponding closed subset of rotations have maximum weight~\cite[Chapter 3.6.2]{GusfieldIrving1989}; the concept revolving around the rotations will be defined shortly.
\citeGIauthors~\cite{GusfieldIrving1989} presented a network flow approach to find such a maximum-weight stable matching in $O(n^2 \cdot w)$~time, where $n$ and $w$ denote the number of agents and the sum of weights of rotations, respectively.
We will show that for our problem the sum of the weights of rotations is bounded by $n$, which is the size of the intersection between our sought stable matching and the target stable matching.

In the remainder of the section, we present necessary notions and the definition of the weights of the rotations and refer to Chapter~3.6.1 by \citeGIauthors{} for further details.

\subsection{Preliminaries}
The approach will heavily utilize the structural properties revolving around stable matchings and rotations of a \textsc{Stable Marriage} instance without ties.

As already observed in the literature, for each \textsc{Stable Marriage} instance~$P$ with two disjoint sets $U$ and $W$ of agents, when operating on the agent set~$U$, the Gale-Shapely algorithm always returns a $U$-optimal stable matching~$M$.
The matching $M$ is \emph{$U$-optimal stable matching} if it is stable and there is no other stable matching~$M'$ such that an agent from $U$ would prefer its partner from $M$ to the partner from $M'$.
Starting from the $U$-optimal stable matching, we can successively eliminate the so-called rotations to obtain further stable matchings.
\begin{definition}[Successor agent, rotations, and rotation elimination]\label{def:rotations}
  Let $P$ be a \textsc{Stable Marriage} instance with two disjoint sets of agents,~$U$ and $W$, and with (possibly) incomplete preferences.
  Given a stable matching $M$ for $P$, for each agent~$u\in U$, we define its \emph{successor~$\sucw_M(u)$} as the \emph{first} (after $M(u)$) agent~$w$ on the preference list of $u$ such that $w$ is matched under $M$ and prefers $u$ to its partner~$M(w)$.

  A sequence~$\rho=((u_{0},w_{0}), (u_1,w_1), \ldots, (u_{r-1}, w_{r-1}))$ of pairs is called a \emph{rotation} if there exists a stable matching $M$ for $P$ such that
  for each $i\in \{0,1,\ldots,r-1\}$ we have $(u_i,w_i)\in U\times W$, $M(u_i)=w_i$, and $\sucw_M(u_i)=w_{i+1}$ (index $i+1$ taken modulo $r$).
  We say that rotation~$\rho$ is \emph{exposed} in~$M$.

  We use the notation~$M/\rho$ to refer to the matching resulting from $M$ by replacing each pair~$\{u_i,w_i\}$ with $\{u_i,w_{i+1}\}$.
  Formally,
  \begin{align*}
    M/\rho = M \setminus \{\{u_i,w_i\} \mid 0\le i \le r-1\} \cup \{\{u_i,w_{i+1}\} \mid 0\le i \le r-1\}.
  \end{align*}
  The transformation of $M$ to $M/\rho$ is called the \emph{elimination of $\rho$ from $M$}.
\end{definition}
  We illustrate the concept of the successor below:
  \begin{align*}
    u\colon \ldots M(u) \ldots \sucw_M(u) \ldots,
    \quad \qquad \sucw_M(u)\colon \ldots u \ldots M(\sucw_M(u)) \ldots
  \end{align*}

  \noindent
  Eliminating a rotation from a stable matching results in another stable matching~\cite[Lemma 2.5.2]{GusfieldIrving1989}.

\begin{definition}[Predecessors of rotations, the rotation poset, and the rotation digraph]\label{def:digraph}
  Let $\pi$ and $\rho$ be two rotations for a \textsc{Stable Marriage} instance~$P$.
  We say that $\pi$ is a \emph{predecessor} of $\rho$,
  written as \emph{$\pi \pred^P \rho$},
  if no stable matching in which $\rho$ is exposed can be obtained from the $U$-optimal stable matching by a sequence of eliminations of rotations without eliminating~$\pi$ first.
  The reflexive closure of the relation~$\pred^P$, denoted as $\predr^P$, defines a partial order on the set of all rotations and is called the \emph{rotation poset} for~$P$.
  We abbreviate the name of a subset of the poset that is closed under predecessors as a \emph{closed subset}.

  An alternative representation of the rotation poset~$\predr(P)$ is through an acyclic directed graph, called \emph{rotation digraph of $P$} and written as \emph{$G(P)$}, whose vertex set is the set of rotations of $P$, and there is a direct arc from rotation $\pi$ to rotation $\rho$ if and only if~$\pi$ precedes~$\rho$ and there is no other rotation~$\sigma$ such that $\pi \pred^P \sigma \pred^P \rho$.
\end{definition}

Finally, let us describe a central result from the literature that relates rotations and stable matchings.

\begin{proposition}[{\cite[Theorem~2.5.7, Lemma~3.3.2]{GusfieldIrving1989}}]\label{prop:rotations+stable-matchings}
  Let $R$ denote the set of all rotations of a preference profile~$P$,
  and let $G(P)$ denote the rotation digraph of $P$. The following holds.
  \begin{enumerate}
    \item\label{rot:closedsubet-stable} A matching~$M$ is a stable matching of $P$ if and only if there is a closed subset of rotations~$R'\subseteq R$ with respect to the precedence relation $\pred^P$ such that $M$ can be generated by taking the $U$\nobreakdash-optimal stable matching and by eliminating the rotations in $R'$ in an order consistent with~$\pred^P$.
    \item\label{rot:runtime} The rotation set~$R$ and the rotation digraph~$G(P)$ can be computed in $O(n^2)$~time.
  \end{enumerate}
\end{proposition}

\todo{Dušan: Say that $R'$ is associated with $M$.}
\todohua{Hua: Where?}

\subsection{From \ISM{} to Finding Maximum-Weight Closed Subset of Rotations}
Given an instance~$I=(\mathcal{P}_1,\mathcal{P}_2,M_1,k)$ of \ISM,
let $M_0$ be the $U$-optimal stable matching of $\mathcal{P}_2$,
let $R_2$ denote the set of rotations for $\mathcal{P}_2$, and let $G_2$ be the rotation digraph for $\mathcal{P}_2$.
Towards finding a stable matching which shall be as close to the input matching~$M_1$ as possible, we assign a weight to each rotation from $R_2$ which shall indicate the benefit of eliminating this rotation.

For each rotation~$\rho\in R_2$ with $\rho=((x_0, y_0), \cdots, (x_{r-1}, y_{r-1}))$ let
\[
  w(\rho)\coloneqq |\{(x_i, y_{i+1}) \mid \{x_i, y_{i+1}\} \in M_1, 0\le i \le r-1\}\}| - |\{(x_i, y_i) \mid \{\{x_i, y_i\} \in M_1, 0\le i \le r-1\}| .
\]
That is, we count the number of pairs in $M_1$ the elimination of $\rho$ introduces minus the number of pairs in $M_1$ we loose when eliminating $\rho$.
By the above definition, we can derive the following.
\begin{lemma}\label{app:lem:SM-rotation-weight}
  If $\rho\in R_2$ is a rotation exposed in a stable matching~$M$ for $\mathcal{P}_2$, then
  \[
  |M_1\cap (M/\rho)| = |M_1\cap M| + w(\rho) \,.
  \]
\end{lemma}
\begin{proof}
  Let $\rho = ((x_0, y_0), \cdots, (x_{r-1}, y_{r-1}))$ be a rotation exposed in the stable matching~$M$.
  In the following, all subscripts~$i+1$ are taken modulo $r$.
  By the definition of $M/\rho$ we have that $M/\rho =\{\{x,y\} \in M \mid
  (x,y) \notin \rho\} \uplus \{\{x_i, y_{i+1}\} \mid 0 \le i \le r-1\}$.
  Thus, we prove the statement by showing the following; note that no $\{x_i,y_i\}$ belongs to $M/\rho$, $0\le i \le r-1$.
  \begin{align*}
    |M_1 \cap (M/\rho)|
    = &  |M_1 \cap \big((M/\rho) \cap M\big)| + |M_1 \cap \big( (M/\rho) \setminus M \big)|& \\
     = &  |\{\{x,y\} \in M_1 \mid \{x,y\} \in M \wedge (x,y) \notin \rho\}| +\\
      & |\{\{x,y\} \in M_1 \mid (x,y)=(x_i, y_{i+1}) \text{ for some }0\le i \le r-1\}| & \\
    = &  |\{\{x,y\} \in M_1 \cap M\}| -
      |\{\{x,y\} \in M_1 \cap M \mid (x,y) \in \rho \} |  + & \\
       &  |\{\{x,y\} \in M_1 \mid (x,y)=(x_i, y_{i+1}) \text{ for some }0\le i \le r-1\}| &\\
    = & |M_1 \cap M| + w(\rho).
  \end{align*}
  Note that the second to last equation holds because $(x,y) \in \rho$ implies that $\{x,y\}\in M$.
\end{proof}

By applying \cref{app:lem:SM-rotation-weight} repeatedly, we obtain the following.
\begin{corollary}\label{app:cor:SM-rotation-weights}
  If $C\subseteq R_2$ is the (unique) closed subset of rotations associated with stable matching~$M$ of $\mathcal{P}_2$,
  then $|M_1\cap M| = |M_1\cap M_0| + \sum_{\rho\in C}w(\rho)$.
\end{corollary}

Since a stable matching~$M_2$ for $\mathcal{P}_2$ with minimum symmetric difference~$\dist(M_1, M_2)$ has the maximum intersection~$M_1 \cap M_2$,
we obtain the following.
\begin{lemma}\label{app:lem:sd-inter}
  Let $M_2$ be a stable matching for $\mathcal{P}_2$ and let $C$ be the associated closed subset of rotations.
  Then, $\dist(M_1, M_2)\le k$ if and only if $\sum_{\rho \in C}w(\rho) \ge \frac{\dist(M_1,M_0) - k }{2}$.
\end{lemma}

\begin{proof}
  By the definition of symmetric difference, we derive the following.
  \begin{align}\label{eq:sd-inter}
    \dist(M_1, M_2) = |M_1| + |M_2| - 2|M_1 \cap M_2|.
  \end{align}
  Since for preferences without ties, all stable matchings match the same set of agents (\cref{prop:IncompleteNoTies}), we have that $|M_0|=|M_2|$.
  Together with \cref{app:cor:SM-rotation-weights}, \eqref{eq:sd-inter} is equivalent to
   \begin{align*}\label{eq:sd-inter}
     \dist(M_1, M_2) &= |M_1| + |M_0| - 2 \left( |M_1 \cap M_0| + \sum_{\rho\in C}w(\rho) \right)\\
     & = \dist(M_1, M_0) - 2 \sum_{\rho\in C}w(\rho).
   \end{align*}
   Our statement follows immediately.
\end{proof}

By the above, our problem reduces to finding a maximum-weight closed subset of rotations.
In the following, we show that the sum of the weights of the rotations is at most $n$ and finding such a subset of rotations can thus be done efficiently.

\begin{lemma}\label{app:lem:SM-max-w-rotations}
  $\sum_{\rho \in R_2}w(\rho)\le |M_1| $ and finding a closed subset of rotations with maximum weight can be done in $O(n^3)$~time.
\end{lemma}
\begin{proof}
  Slightly abusing the intersection notation,
  for each rotation~$\rho=((x_0,y_0), \cdots, (x_{r-1}, y_{r-1})) \in R_2$ we define
  $\rho \cap M_1 \coloneqq \{(x_i, y_{i+1}) \mid \{x_i, y_{i+1} \in M_1 \}\}$  ($i+1$ is taken modulo $r$).

  Next, observe that by one of the statement in \cite[Lemma 3.2.1]{GusfieldIrving1989} implies that for each pair~$(x, y)$ there exists at most one rotation~$\rho$ such that $(x,y) \in \rho \cap M_1$.
  \todo{Dušan: Shall we write the Lemma 3.2.1 in preliminaries here for readers convenience?}
  \todohua{Hua: Lemma 3.2.1 contains more statements and uses some notions that I don't want to introduce.}
  Moreover, by definition, we have that $w(\rho) \le |\rho \cap M_1|$.
  Summarizing, we have that
  \begin{align*}
    \sum_{\rho\in R_2} w(\rho) \le \sum_{\rho \in R_2}|\rho\cap M_1| = |\bigcup_{\rho\in R_2} (\rho \cap M_1)| \le |M_1|.
  \end{align*}

  \todo{Dušan: I think the rest of the proof is not needed---Theorem 3.6.2 is already mentioned in the section introduction, however, we maybe want to add its statement for reader's convenience}
  \todohua{Again, Theorem 3.6.2 uses lots of notions and technical statements that are not relevant to us, and I tried to avoid this.\newline I think the rest (maybe thinned a little) is relevant for the sake of completeness of the proof.}
  \citeGIauthors~\cite[Theorem 3.6.2]{GusfieldIrving1989} showed that finding a maximum-weight closed subset of rotations can be reduced to finding a minimum $s$-$t$ cut in a specific flow network which features the precedence relation of the rotations,
  where the numbers of vertices and arcs are in $O(n^2)$ and the minimum cost of an $s$-$t$ cut is bounded by the sum of the weights of the rotations.
  The latter problem can be solved in $O(|E|\cdot w)$ time (by using Ford-Fulkerson's algorithm), where $|E|$ denotes the number of arcs in the network and $w$ is the cost of the minimum $s$-$t$ cut.
  Since $|E|\in O(n^2)$ and since $w\le \sum_{\rho\in R_2}w(\rho)$, we can find a maximum-weight closed subset of rotations in $O(n^2 \cdot |M_1|) = O(n^3)$~time.
\end{proof}







\end{document}